\documentclass[11pt]{article}
\usepackage{enumitem}
\makeatletter
\newcommand{\mylabel}[2]{#2\def\@currentlabel{#2}\label{#1}}
\makeatother
\usepackage{braket}
\usepackage{amsfonts,amssymb}
\usepackage{color}
\usepackage{amsthm, amssymb, latexsym, amsmath,  %showkeys
}
\usepackage[margin=1.2 in]{geometry}
\usepackage[ngerman,english]{babel}
\usepackage[utf8x]{inputenc}
\usepackage[OT2,T1]{fontenc}

\newtheorem{theorem}{Theorem}[section]

\newtheorem{lemma}[theorem]{Lemma}

\newtheorem{corollary}[theorem]{Corollary}

\theoremstyle{definition}

\newtheorem{remark}{Remark}

\newcommand{\DETAILS}[1]{}

\newcommand{\re}{{\rm{Re\, }}}
\newcommand{\im}{{\rm{Im\, }}}
\renewcommand{\Im}{\mathrm{Im}}
\renewcommand{\Re}{\mathrm{Re}}

\newcommand{\C}{\mathbb{C}}
\newcommand{\R}{\mathbb{R}}

\newcommand{\N}{\mathbb{N}}

\newcommand{\cA}{{\cal{A}}}

\newcommand{\cH}{{\cal{H}}}

\newcommand{\cO}{{\cal{O}}}

\newcommand{\cE}{\mathcal{E}}
\renewcommand{\i}{\mathrm{i}}

\newcommand{\vphi}{{\varphi}}
\newcommand{\e}{{\mathrm{e}}}
\newcommand{\scp}[2]{\left\langle #1\text{,}\, #2\right\rangle}
\newcommand{\tr}[1]{\mathrm{Tr}#1}
\newcommand{\ima}[1]{\mathrm{Im}#1}
\begin{document}
	
	\title{Derivation of the Hartree equation for compound Bose gases in the mean field limit}
	
	\author{Ioannis Anapolitanos\thanks{Insitute for Analysis,
			Karlsruhe Institute of Technology, Englerstra\ss{}e 2,  76131 Karlsruhe, Germany.}, Michael Hott\thanks{ Department of Mathematics,
			University of Texas at Austin, 2515 Speedway,
			Austin, TX 78712, USA.},
	 Dirk Hundertmark\thanks{Insitute for Analysis,
	 	Karlsruhe Institute of Technology, Englerstra\ss{}e 2,  76131 Karlsruhe, Germany.} }
	
	\bigskip
	
	\bigskip
	
	\bigskip

	%%%%%%%%%%%%%%%%%%%%%%%%%%%%%%%%%%%%%%%%%%
	\bigskip
	\bigskip

	\maketitle
	
	\begin{abstract}
	   We consider mixtures of Bose gases of different species. We prove that in the mean field limit and under suitable conditions on the initial condition a system composed of two Bose species can be effectively described by a system of coupled Hartree equations. Moreover, we derive quantitative bounds on the rates of convergence of the reduced density matrices in Sobolev trace norms. We treat both the non-relativistic case in the presence of an external magnetic field $A\in L^2_{\text{loc}}(\R^3;\R^3)$ and the semi-relativistic case.
	\end{abstract}
	
\section{Introduction}

We consider a system of $N_1$ bosons of one species and $N_2$ bosons of another species all interacting with each other, described by the Hamiltonian 
\begin{equation}\label{manybodyham}
\begin{split}
H_{N}\;=\;&\sum_{i=1}^2\sum_{j=1}^{N_i}\left[S_j^{(i)} + \frac{1}{N_i-1} \sum_{k=j+1}^{N_i} u^{(i,i)}(x_{j}^{(i)}-x_{k}^{(i)})\right]\\& + \frac{1}{m(N_1,N_2)} \sum_{i=1}^{N_1} \sum_{j=1}^{N_2} u^{(1,2)}(x_{i}^{(1)}-x_{j}^{(2)})
\end{split}
\end{equation}
with $N:=(N_1,N_2)\in\N^2$. Here $u^{(i,j)}$ plays the role of a two-body interaction potential between bosons of the $i^{th}$ and $j^{th}$ species, $m(N_1,N_2)$ plays the role of a mean of $N_1$ and $N_2$ and will be determined below. We will choose it to ensure that the energy scales, up to a factor of order of magnitude $\cO(1)$, between 
$\min(N_1,N_2)$ and $\max(N_1,N_2)$. The position of the $j^{th}$ particle of the $i^{th}$ species is denoted by $x_j^{(i)} \in \R^3$ and $S_j^{(i)}$ denotes its kinetic energy, possibly also including an external potential. In our setup, we are interested in the cases 
\begin{equation*}
	S_j^{(i)}\in\left\{\frac{1}{2m_i}D_{A,{x_{j}^{(i)}}}^2,\sqrt{m_i^2-\Delta_{x_j^{(i)}}}\right\},
\end{equation*}
where $D_{A}^2:=(-\i\nabla+A)^2$ denotes the magnetic Laplacian for a magnetic field $A:\R^3 \to \R^3$ and $m_i$ is the mass of a particle of the $i^{th}$ spieces. In addition, we assume that the interaction potentials are Yukawa potentials. More precisely, there are $\mu^{(1,1)}, \mu^{(2,2)}, \mu^{(1,2)} \geq0$, and $\lambda^{(1,1)}, \lambda^{(2,2)}, \lambda^{(1,2)} \in\R$, s.t.
\begin{equation}\label{eq:yukawa}
u^{(i,j)}(x)=\lambda^{(i,j)} \frac{\e^{-\mu^{(i,j)}|x|}}{|x|}\qquad\mbox{for\ } 1\leq i\leq j\leq 2.
\end{equation}
	It is well-known that imposing $A\in L^2_{\text{loc}}(\R^3;\R^3)$ leads to a self-adjoint operator $H_N$, see \cite{simon-schroedinger-forms}. This is based on the fact that $D_A^2$ can be realized as a self-adjoint operator as we explain in Appendix \ref{app:wellposed} and on the Kato-Rellich Theorem, see for example Theorem X.12 in \cite{RS2}. In  Appendix \ref{app:self-adjoint}, we also show that in the semi-relativistic case the operator  $H_N$ is  self-adjoint provided that the assumptions \ref{itm:sr} and \ref{itm:mf} below  hold. 

\par

In this work, we are interested in the case where both boson types initially are in a BEC and we want to investigate the dynamics of such a system. More precisely, we study the Cauchy problem
\begin{equation}\label{totalpde}
\begin{split}
\mathrm{i}\partial_t\Psi_{N,t}\;&=\;H_N\Psi_{N,t}\\
\Psi_{N,0}\;&=\;\psi_0^{\otimes N_1} \otimes  \vphi_0^{\otimes N_2}
\end{split}
\end{equation} 
with $\|\psi_0\|_{2}=\|\vphi_0\|_{2}=1$ on the Hilbert space
\begin{equation*}
\mathcal{H}_N\;:=\;L^2(\R^3)^{\otimes_s N_1}\otimes L^2(\R^3)^{\otimes_s N_2},
\end{equation*}
where $L^2(\R^3)^{\otimes_s N_i}$ denotes the $N_i$-fold symmetric tensor product of $L^2(\R^3)$.
We will show that under certain assumptions, the system remains approximately a product of Bose condensates, where the terms of the product are described by a system of coupled Hartree equations.  We start by heuristically deriving the evolution equation of the effective field. Using the ansatz $\Psi_{N,t} = \psi_t^{\otimes N_1} \otimes  \vphi_t^{\otimes N_2}$, we obtain
\begin{equation}\label{eq:chndef} 
\begin{split}
&\big\langle\Psi_{N,t},H_N\Psi_{N,t}\big\rangle\; =\;
N_1\langle \psi_t,  S_1^{(1)}  \psi_t \rangle +  N_2\langle \vphi_t,S_1^{(2)}  \vphi_t \rangle +  \frac{N_1}{2} \langle\psi_t^{\otimes2},  u^{(1,1)}(x_1^{(1)}-x_2^{(1)})  \psi_t^{\otimes2} \rangle \\
&\phantom{=\;}+\frac{N_2}{2}\langle \vphi_t^{\otimes2},  u^{(2,2)}(x_1^{(2)}-x_2^{(2)})\vphi_t^{\otimes2}\rangle
+\frac{N_1N_2}{m(N_1,N_2)} \langle  \psi_t \otimes  \vphi_t,  u^{(1,2)}(x_1^{(1)}-x_1^{(2)})  \psi_t \otimes \vphi_t \rangle\\
&=:\; \mathcal{H}_N(\overline{\psi_t},\overline{\vphi_t},\psi_t,\vphi_t).
\end{split}
\end{equation}

This leads to the dynamics 
\begin{equation}\label{hartree1}
\begin{split}
\i \partial_t \psi_t=N_1^{-1}\partial_{\overline{\psi_t}} \mathcal{H}_N(\overline{\psi_t},\overline{\vphi_t},\psi_t,\vphi_t)&=  S_1^{(1)} \psi_t +  u^{(1,1)} * |\psi_t|^2 \psi_t + \frac{N_2}{m(N_1,N_2)}u^{(1,2)} * |\vphi_t|^2 \psi_t\\
\i \partial_t \vphi_t=N_2^{-1} \partial_{\overline{\vphi_t}} \mathcal{H}_N\left(\overline{\psi_t},\overline{\vphi_t},\psi_t,\vphi_t\right)&= S_1^{(2)} \vphi_t +  u^{(2,2)} * |\vphi_t|^2 \vphi_t + \frac{N_1}{m(N_1,N_2)} u^{(1,2)} * |\psi_t|^2 \vphi_t\\
(\psi_t,\vphi_t)_{t=0}&=(\psi_0,\vphi_0).
\end{split} 
\end{equation}
The factor $N_i^{-1}$ arises due to the fact that an effective equation for the $i^{th}$ species describes the averaged dynamics. To formulate it differently, an effective equation describes the motion of a typical particle. A formal computation shows that $\mathcal{H}_N\left(\overline{\psi_t}, \overline{\vphi_t}, \psi_t, \vphi_t\right)$ is conserved under the dynamics \eqref{hartree1}. In Appendix \ref{app:wellposed} we discuss how this can be rigorously shown under the assumptions below.

To obtain limiting dynamics of the coupled fields, we have to impose $\lim_{N_1,N_2\to\infty}\frac{N_j}{m(N_1,N_2)}\in\R^+_0$. Note that at least one of these limits is positive since $m(N_1,N_2)\leq\cO(\max(N_1,N_2))$. If the other were zero, we would obtain dynamics of that field independent of the presence of the other species. We will assume that both limits exist and are positive. This gives $m(N_1,N_2)=\cO(\sqrt{N_1N_2})(=\cO(N_1)=\cO(N_2))$, for convenience we choose $m(N_1,N_2)=\sqrt{N_1N_2}$. For simplicity, we also consider only those $(N_1,N_2)\in\N^2$ for which
\begin{align}
R\;&:=\;\sqrt{\frac{N_1}{N_2}}\label{eq:ratiofix}\in\R,
\end{align}
is fixed even though it would be sufficient to impose $\sqrt{\frac{N_1}{N_2}}\to R$ as $N_1,N_2\to\infty$. This leads to the dynamics
\begin{equation}\label{hartree}
\begin{split}
	\i \partial_t \psi_t=  S^{(1)} \psi_t +  u^{(1,1)} * |\psi_t|^2 \psi_t & + \frac{1}{R} u^{(1,2)} * |\vphi_t|^2 \psi_t\\
	\i \partial_t \vphi_t= S^{(2)} \vphi_t +  u^{(2,2)} * |\vphi_t|^2 \vphi_t & + R u^{(1,2)} * |\psi_t|^2 \vphi_t\\
	(\psi_t,\vphi_t)_{t=0}&=(\psi_0,\vphi_0).
\end{split} 
\end{equation}
where 
\begin{equation}\label{eq:kinetic}
	S^{(i)}\in\left\{\frac{1}{2m_i}D_{A}^2,\sqrt{m_i^2-\Delta}\right\},
\end{equation}
is the kinetic energy of the $i^{\rm th}$ species. 

The problem of compositions of Bose-Einstein condensates has been investigated in the physics literature, see for example \cite[Chapter 21]{PS}, \cite{CNPV}, \cite{MBGCW}, \cite{MMRRI}, and references therein. Mathematically, this problem was first studied in \cite{He}, where it was observed that the method developed by Pickl in \cite{Pi} can be used to treat mixtures of Bose-Einstein Condensates. There is a further general result by Olgiati and Michelangeli see \cite{MO},  where the problem is studied for a quite general class of potentials and assumptions on the initial data. Just a few days ago, there appeared a work by Olgiati, see \cite{O}, where the Gross-Pataevskii regime is also studied. After stating our main results, we will explain in the introduction what is new in our results.

\subsection{General assumptions and global well-posedness of \eqref{hartree}} 

Throughout this work, we assume the \emph{mean-field} condition 
\begin{enumerate}
	\item[\mylabel{itm:mf}{{\rm (}MF{\rm )}}] the scaling equations \eqref{eq:ratiofix} and $m(N_1,N_2)=\sqrt{N_1N_2}$ hold. 
\end{enumerate}
Depending on whether we consider the magnetic or semi-relativistic case, we also assume
the following regularity conditions.
\begin{enumerate}
	\item[\mylabel{itm:rm}{(RM)}] $A \in L_{\text{loc}}^2(\R^3)$,  $\psi_0,\vphi_0\in H_A^1(\R^3)$, $\|\psi_0\|_{L^2}=\|\vphi_0\|_{L^2}=1$ and $S_j^{(i)}=\frac{1}{2m_i}D_{A,{x_{j}^{(i)}}}^2$ (magnetic case)
	\item[\mylabel{itm:rs}{(RS)}] $\psi_0,\vphi_0\in H^1(\R^3)$, $\|\psi_0\|_{L^2}=\|\vphi_0\|_{L^2}=1$ and $S_j^{(i)}=\sqrt{m_i^2-\Delta_{x_j^{(i)}}}$ (semi-relativistic case).
\end{enumerate}
As we will see in the next theorem, the following condition is sufficient for \eqref{hartree} to be globally well-posed in the semi-relativistic case.
\begin{itemize}
	\item[\mylabel{itm:sr}{(SR)}]
	$(\lambda^{(1,2)}_{-})^2\;<\;\left(\frac{4}{\pi}-\lambda_{-}^{(1,1)}\right)\left(\frac{4}{\pi}-\lambda^{(2,2)}_{-}\right),\quad \lambda^{(1,1)}_{-},\lambda^{(2,2)}_{-}\;<\;\frac{4}{\pi}$.
\end{itemize}
In here $\lambda_{-}^{(i,j)}:=-\min(0,\lambda^{(i,j)})$ refers to the negative-part of $\lambda^{(i,j)}$.

\begin{theorem}\label{thm:wellposed}
	\begin{enumerate}[label=\textnormal{(\arabic*)}]
		\item Assume \ref{itm:rm}. Then the Cauchy problem \eqref{hartree} is globally well-posed. More precisely, for all $T>0$ there is a unique mild solution
	\begin{equation}\label{eq:CC1}
	(\psi, \vphi) \in C\left([0,T), H_A^1(\R^3)\times H_A^1(\R^3)\right) \cap  C^1\left([0,T),  H_A^{-1}(\R^3) \times H_A^{-1}(\R^3) \right)
	\end{equation}
	of \eqref{hartree}, which continuously depends on $(\psi_0, \vphi_0)$. 
	\item Assume \ref{itm:rs} and \ref{itm:sr}. Then the Cauchy problem \eqref{hartree} is globally well-posed. More precisely, for all $T>0$ there is a unique mild solution
	\begin{equation}\label{eq:CCr}
	(\psi, \vphi) \in C\left([0,T), H^1(\R^3)\times H^1(\R^3)\right) \cap  C^1\left([0,T), L^2(\R^3)\times L^2(\R^3)\right)
	\end{equation}
	of \eqref{hartree}, which continuously depends on $(\psi_0, \vphi_0)$. 
	\end{enumerate}
\end{theorem}
Theorem \ref{thm:wellposed} can be proven adopting the methods of \cite{Ca} in the magnetic case and of \cite{Le} in the semi-relativistic case. For the magnetic case we also refer to \cite{Mi}. In Appendix \ref{app:wellposed} we outline the proof and explain the necessary modifications.

%\begin{remark}
%Assume \ref{itm:a1}.  In addition to that, assume in the semi-relativistic case that the negative-part of the interaction terms of $H_N$ is $\sum_{i=1}^2\sum_{j=1}^{N_i}S_j^{(i)}$-bounded with relative bound less than 1, which is the case if \ref{itm:sr} holds. Then $H_N$ is a self-adjoint operator with form domain $H^\frac12(\R^{3(N_1+N_2)})$. Indeed, using that $u^{(i,j)}$ is $\sqrt{1-\Delta}$ bounded for all $i,j$, it follows that $H_{N}$ defines a quadratic form that is bounded from below and closed on $H^\frac12(\R^{3(N_1+N_2)})\cap\mathcal{H}_N$. In the non-relativistic case, it follows from Hardy's Inequality and By the Kato-Rellich Theorem that $H_N$ is self-adjoint with form domain $H^1(\R^{(3(N_1+N_2))})\cap\mathcal{H}_N$.
%\end{remark}

\subsection{Main results}
Before we state our main results, we introduce some convenient notation. Let $\Psi_N$ be the solution of \eqref{totalpde} and $(\psi,\vphi)$ be the solution of \eqref{hartree} in Theorem \ref{thm:wellposed}, depending on the particularly considered case. Then we define for all $k\leq N_1,l\leq N_2$, $t,\theta \in \R$ 
\begin{equation} \label{eq:s-k-l-theta}
\begin{split}
\gamma_{N,t}^{(k,l)}\;&:=\;\tr_{k+1,\ldots,N_1+N_2-l} \ket{\Psi_{N,t}}\bra{\Psi_{N,t}},\\
P^{(k,l)}_t\;&:=\;(\ket{\psi_t}\bra{\psi_t})^{\otimes k}\otimes (\ket{\vphi_t}\bra{\vphi_t})^{\otimes l},\\
S_{k,l,\theta}\;&:=\;\sum_{i=1}^k (1+S_i^{(1)})^{\theta}+\sum_{i=1}^l (1+S_i^{(2)})^{\theta}.
\end{split}
\end{equation}
Here $\tr_{k+1,\ldots,N_1+N_2-l}$ is shorthand for taking the  partial trace with respect to the last $N_1-k$ particles of the first species and the first $N_2-l$ particles of the second species, that is, for any bounded operators $A$ on $L^2(\R^3)^{\otimes_s k}$ and $B$ on $L^2(\R^3)^{\otimes_s l}$ we have 
\begin{align*}
	\tr_{L^2(\R^3)^{\otimes_s k}\otimes L^2(\R^3)^{\otimes_s l}} \Big((A\otimes B) \gamma_{N,t}^{(k,l)}\Big) = \tr_{\mathcal{H}_N}\Big( \big((A\otimes I^{\otimes (N_1-k)})\otimes (I^{\otimes (N_2-l)}\otimes B)\big) \ket{\Psi_{N,t}}\bra{\Psi_{N,t}}\Big).
\end{align*}
The operator $\gamma_{N,t}^{(k,l)}$ is often referred to as the reduced density matrix of the state $\ket{\Psi_{N,t}}\bra{\Psi_{N,t}}$. We are now ready to state our main results.
\begin{theorem}[Magnetic case]\label{thm:magnetic}
Assume \ref{itm:mf} and \ref{itm:rm}, fix $k,l\in\N_0$, not both identically zero, and let $N_1,N_2\in\N$ with  $N_1\ge k$, $N_2\ge l$ and assume that, for some fixed $R>0$,  $R=\sqrt{N_1/N_2}$.  
\begin{enumerate}[label=\textnormal{(\arabic*)}]
	\item There are constants $C,D>0$ depending only on $R$, $\vphi_0, \psi_0$ such that for any 
	      $\theta\in[0,1)$ and $t > 0$ we have 
	\begin{align}\label{ineq:magmain}
	\tr{\left|S_{k,l,\theta}^{\frac12}\left(\gamma_{N,t}^{(k,l)} -P^{(k,l)}_t\right)S_{k,l,\theta}^{\frac12}\right|}& \leq C\frac{(k+l)^{\frac{3-\theta}2} \e^{Dt}}{N_1^{(1-\theta)/2}}.
	\end{align}
	\item Let $\theta=1$. Then for all $t>0$ we have
	\begin{equation}\label{conv:magen}
	\lim_{\substack{N_1,N_2\to\infty\\ N_1/N_2=R^2}}\tr{\left|S_{k,l,1}^{\frac12}\left(\gamma_{N,t}^{(k,l)} -P^{(k,l)}_t\right)S_{k,l,1}^{\frac12}\right|}\;=\;0.
	\end{equation}
\end{enumerate}

\end{theorem}

\begin{theorem}[Semi-relativistic case]\label{thm:semirel}
	Assume \ref{itm:mf}, \ref{itm:rs} and \ref{itm:sr},  fix $k,l\in\N_0$, not both identically zero, and let $N_1,N_2\in\N$ with  $N_1\ge k$, $N_2\ge l$ and assume that, for some fixed $R>0$,  $R=\sqrt{N_1/N_2}$. Then there are constants $C,D_1,D_2>0$ depending only on $R$, $\vphi_0, \psi_0$ such that for any 
	$\theta\in[0,1)$ and $t > 0$ we have 
	\begin{align}\label{ineq:srmain}
	\tr{\left|S_{k,l,\theta}^{\frac12}\left(\gamma_{N,t}^{(k,l)} -P^{(k,l)}_t\right)S_{k,l,\theta}^{\frac12}\right|}& \leq C\frac{(k+l)^{\frac{3-\theta}2} \e^{D_1\e^{D_2 t}}}{N_1^{(1-\theta)/2}}.
	\end{align}
\end{theorem}
For results of convergence in Sobolev trace norms in the case of one species of particles we refer to \cite{Lu}, \cite{MiS}, \cite{MPP} and \cite{AH}.
The following corollary gives a physical interpretation of inequality \eqref{ineq:magmain}  in terms of expectation values of quantum observables. There is a similar interpretation of \eqref{conv:magen} and \eqref{ineq:srmain}.
\begin{corollary}\label{cor:observ}
Assume the assumptions of Theorem \ref{thm:magnetic}. Let $\cA$ be a self-adjoint operator on $L^2(\R^3)^{\otimes_S k} \otimes L^2(\R^3)^{\otimes_S l}$ such that there is a $\theta \in [0,1)$ so that $S_{k,l,\theta}^{-{1/2}}\cA S_{k,l,\theta}^{-{1/2}}$ can be extended to a bounded operator. Then for all $t>0$ we have
	\begin{equation}\label{eq:cor}
	\begin{split}
		\left|\scp{\Psi_{N,t}}{I^{\otimes (N_1-k)} \otimes  \cA \otimes I^{\otimes (N_2-l)} \Psi_{N,t}} - \scp{\psi_t^{\otimes k}\otimes\vphi_t^{\otimes l}}{ \cA  \left(\psi_t^{\otimes k}\otimes\vphi_t^{\otimes l}\right)}\right| \\\leq C\frac{(k+l)^{\frac{3-\theta}2} \e^{Dt}}{N_1^{(1-\theta)/2}} \left\|S_{k,l,\theta}^{-\frac12}\cA S_{k,l,\theta}^{-\frac12}\right\|,
	\end{split}
	\end{equation}
where $\|S_{k,l,\theta}^{-{1/2}}\cA S_{k,l,\theta}^{-{1/2}}\|$ is the operator norm of $S_{k,l,\theta}^{-{1/2}}\cA S_{k,l,\theta}^{-{1/2}}$ and $C,D$ are the same constants as in Theorem \ref{thm:magnetic}.
\end{corollary}

\begin{proof}
 This follows from Theorem \ref{thm:magnetic} using a simple duality argument: Using the definition of the reduced density matrix and the cyclicity of the trace we have  
 \begin{align*}
 	\text{l.h.s of } \eqref{eq:cor} &= \left|\tr\Big( \cA\big(\gamma_{N,t}^{(k,l)} -P^{(k,l)}_t\big) \Big) \right|
 	= \left|\tr\Big( S_{k,l,\theta}^{-\frac12}\cA S_{k,l,\theta}^{-\frac12} S_{k,l,\theta}^{\frac12}\big(\gamma_{N,t}^{(k,l)} -P^{(k,l)}_t\big) S_{k,l,\theta}^{\frac12}\Big) \right| \\
 	&\le \left\|S_{k,l,\theta}^{-\frac12}\cA S_{k,l,\theta}^{-\frac12}\right\| 
 		\tr{\left|S_{k,l,\theta}^{\frac12}\left(\gamma_{N,t}^{(k,l)} -P^{(k,l)}_t\right)S_{k,l,\theta}^{\frac12}\right|}
 \end{align*}
 since the space of bounded operators is the dual of the trace class operators, see \cite[Theorem VI.26]{RS1}.  The bound \eqref{eq:cor} now follows from \eqref{ineq:magmain}. 
\end{proof}

\begin{remark}
	\begin{enumerate}[label=\arabic*),ref=\arabic*, leftmargin=1.4em,itemsep=0mm]
		\item The bound corresponding to \eqref{eq:cor} in the  semi-relativistic case would follow from the same argument as above and \eqref{ineq:srmain}. 
		\item Note that throughout this work, we do not specify the scalar product nor the trace, but we always refer to the ambient $L^2$-space without the Bosonic symmetry. By $\|.\|_2$ we will denote the $L^2$-norm.
		\item Our proof shows that in the case $\theta=0$, the factor $(k+l)^{3/2}$ on the r.h.s. of \eqref{ineq:magmain} respectively \eqref{ineq:srmain} can be improved to $(k+l)^\frac{1}{2}$, see inequality \eqref{ineq:magmain0} respectively \eqref{ineq:srmain0}  below.
		\item One may extend Theorems \ref{thm:magnetic}  and \ref{thm:semirel} to the case of $r\in\N$ distinct Bose particle species if one suitably modifies assumption \ref{itm:sr}.
		\item To our best of knowledge mean field dynamics with the weak assumption $A \in L^2_{\text{loc}}(\R^3)$ for the magnetic potential has not been studied in the previous literature even in the case of one species of Bosons. In particular, our assumptions on the magnetic field are much weaker than the assumptions in \cite{Lu}, where $A$ has to be infinitely many times differentiable and the derivatives have to be bounded. In addition, we assume only $H_A^1$-regularity of the initial data in the magnetic case, which is weaker than the $H_A^3$-regularity  assumption in Theorem 1.2 of \cite{Lu}. Moreover, our method which yields convergence in Sobolev trace norms is considerably simpler than the ones in \cite{Lu} and \cite{MiS}. Even with stronger regularity assumptions of the initial data and in the case of one species of particles, the estimates \eqref{ineq:magmain} and \eqref{ineq:srmain} do not immediately follow from the methods in \cite{MiS} and \cite{Lu}. The reason is that if for example $\vphi_0 \in H^2$ then there is only a superexponential bound in time for the growth of $\|\vphi_t\|_{H^2}$ in the semi-relativistic case, see Lemma 3 in  \cite{Le}. Such bound would lead to a time dependence which is worse than in the right hand side of \eqref{ineq:srmain} even if the $N$-dependence were better. However, we would also like to stress that unlike \cite{Lu} and \cite{MiS} we were not able to find, with our methods, any rates of convergence in the case of the energy trace norm, even assuming additional regularity of the initial data. Furthermore, in the semi-relativistic case we do not have convergence in the energy trace norm, since in this case the Yukawa interaction is critical. 
		
	\item At this point, we would like to explain what is new in our results in comparison with the previous results in this direction and especially in comparison with \cite{MO}. We first want to emphasize the fact that we have convergence in stronger trace norms and in the magnetic case in the energy trace norm as well. In addition, in the semi-relativistic case, our Theorem \ref{thm:semirel} even for $\theta=0$ does not immediately follow from Theorem 2.1 in \cite{MO}. This  relies on the fact that our assumptions do not imply assumption \textit{(A2)} and \textit{(A4)}, nor assumption \textit{(A2')}, in \cite{MO}. The reason is that since the Coulomb potential is in $L^r+L^\infty$, as required in assumption \textit{(A2)} in \cite{MO}, only for $r<3$, the initial data is required to be in $L^p\cap H^{1/2}\not\subseteq H^1$ with $p>6$. Assumption \textit{(A2')}) in \cite{MO} is not met by the Coulomb potential. In his very recent work \cite{O} that appeared almost simultaneously with ours, Olgiati improves the $L^p\cap H^{1/2}$-assumption on the initial data in the semi-relativistic case, still only in the case $\theta=0$, to $H^1$ and the assumptions also allow the Coulomb potential. Furthermore, we point out that from our well-posedness result in the magnetic case, our result for $\theta=0$ would follow from the result of \cite{MO}.
		
	\end{enumerate}
	
\end{remark}

\section{Proof of Theorems \ref{thm:magnetic} and \ref{thm:semirel}}

We first prove the theorems in the case $\theta=0$, adopting the methods of \cite{Pi} and \cite{KP} that were devepoled for the case of one species of particles. We then prove them in the case $\theta \in (0,1)$ using methods in \cite{AH}. In the end, we prove Theorem \ref{thm:magnetic} in the case $\theta=1$. In this case, the proof in \cite{AH} is not easily generalizable and there are significant modifications in the argument.

\subsection{The case $\theta=0$.}\label{sec:theta-0}
    We note that, as mentioned in the introduction, we are not the first that observed, that with the methods of Pickl the case $\theta=0$ can be handled. However, for convenience of the reader we will outline it.
	Like in [Pi] we define for $h \in L^2(\mathbb{R}^3)$ the projection $p_{i,j}^h:\cH_N\to\cH_N$ by
	\begin{equation*}
	p_{i,j}^h\Phi(x)\;:=\;h(x_{i}^{(j)})\int \overline{h(y_{i}^{(j)})}\Phi(x_1^{(1)},\ldots,y_{i}^{(j)},\ldots,x_{N_2}^{(2)})\mathrm{d} y_{i}^{(j)}
	\end{equation*}
	for any $x=(x_1^{(1)},\ldots,x_{N_1}^{(1)},x_1^{(2)},\ldots,x_{N_2}^{(2)})$. Let as well 
	\begin{equation}\label{eq:qij}
	q_{i,j}^h  = 1- p_{i,j}^h.
	\end{equation}
	We define as in \cite{Pi}	
	\begin{equation}\label{an1}
	a_{N}^{(1)}(\Psi_{N,t},\psi_t):= \langle \Psi_{N,t}, q_{1,1}^{\psi_t} \Psi_{N,t} \rangle =\|q_{1,1}^{\psi_t}\Psi_{N,t}\|_2^2
	\end{equation}
	and
	\begin{equation}\label{an2}
	a_{N}^{(2)}(\Psi_{N,t},\vphi_t):= \langle \Psi_{N,t}, q_{1,2}^{\vphi_t} \Psi_{N,t} \rangle=\|q_{1,2}^{\vphi_t}\Psi_{N,t}\|_2^2.
	\end{equation}
	Note that the initial condition in \eqref{totalpde} yields 
	\begin{equation}\label{eq:aninitial}
	a_{N}^{(1)}(\Psi_{N,0},\psi_0)\;=\;a_{N}^{(2)}(\Psi_{N,0},\vphi_0)\;=\;0.
	\end{equation}
	We claim that 
	\begin{equation}\label{est:a}
	\tr\left|\gamma_{N,t}^{(k,l)}-P^{(k,l)}_t \right| \leq \sqrt{ 8 (k a_{N}^{(1)}(\Psi_{N,t},\psi_t)+l a_{N}^{(2)}(\Psi_{N,t},\vphi_t) )}.
	\end{equation}
	Indeed, one can show with the variational characterization of eigenvalues that the operator on the left hand side has at most one negative eigenvalue, because $P^{(k,l)}_t$ is a rank-one projection. It follows that its trace norm is less than or equal to twice its operator norm  and in particular twice its Hilbert-Schmidt norm. This argument is known, see for example \cite{RSc}. As a consequence, we have
	\begin{equation}\label{aprest}
	\begin{split}
	\tr{\left|\gamma_{N,t}^{(k,l)}-P^{(k,l)}_t \right|}\; &\leq\;2 \sqrt{ \tr\left(\gamma_N^{(k,l)}-P^{(k,l)}_t \right)^2 }\\
	&\leq\;\sqrt{8 \scp{\Psi_{N,t}}{ (1- p_{N_1-k+1,1}^{\psi_t} \dots  p_{N_1,1}^{\psi_t} p_{1,2}^{\vphi_t} \dots p_{l,2}^{\vphi_t}) \Psi_{N,t} }}.
	\end{split}
	\end{equation}
	For clarity, call the projections appearing on the right hand side $p_1 \dots p_{k+l}$. Observe that $1- p_1 \dots p_{k+l}=\sum_{j=1}^{k+l}(1-p_j) p_1... p_{j-1} \leq \sum_{j=1}^{k+l} q_j $, which together with \eqref{aprest}, \eqref{an1}, \eqref{an2}, and the symmetry of $\Psi_{N,t}$ implies estimate \eqref{est:a}.\footnotemark{\footnotetext{We would like to thank Marcel Griesemer for pointing out this simplified argument to us.}}
	\par Our goal is to show that there exists a positive function $f$ such that for all $N=(N_1,N_2)\in\N^2$ with $\frac{N_1}{N_2}=R^2$ we have
	\begin{equation}\label{Gron1}	
	\partial_t( a_{N}^{(1)}(\Psi_{N,t},\psi_t)+a_{N}^{(2)}(\Psi_{N,t},\vphi_t))\leq f(t) \left(a_{N}^{(1)}(\Psi_{N,t},\psi_t) + a_{N}^{(2)}(\Psi_{N,t},\vphi_t) + \frac{1}{N_1}\right),
	\end{equation}
	where in the magnetic case there exists $D>0$ such that 
	\begin{equation}\label{est:fmag}
	f(t) \leq D, \quad \forall t>0
	\end{equation}
	 and in the semi-relativistic case there exist $D_1,D_2>0$ such that 
	 \begin{equation}\label{est:fsr}
	 f(t) \leq D_1 e^{D_2t}, \quad \forall t>0.
	 \end{equation}
	 Assuming \eqref{Gron1} for the moment, we can  apply Gronwall's inequality and \eqref{eq:aninitial} to show
	 \begin{equation}\label{est:aN}
	a_{N}^{(1)}(\Psi_{N,t},\psi_t) + a_{N}^{(2)}(\Psi_{N,t},\vphi_t) \leq \frac{e^{\int_0^t f(\tau)\mathrm{d} \tau}}{N_1}
	\end{equation}
	 which, together with \eqref{est:a}  and \eqref{est:fmag}, implies that in the magnetic case  for all $t>0$, $0\leq k\leq N_1$, $0\leq l\leq N_2$ we have
	\begin{align}\label{ineq:magmain0}
	\tr{\left|\gamma_{N,t}^{(k,l)} -P^{(k,l)}_t\right|}& \leq \sqrt{8} \frac{(k+l)^{\frac{1}2} \e^{Dt/2}}{N_1^{1/2}}.
	\end{align}
Similarly, in the semi-relativistic case it follows that for all $t>0$, $0\leq k\leq N_1$, $0\leq l\leq N_2$ we have
\begin{align}\label{ineq:srmain0}
\tr{\left|\gamma_{N,t}^{(k,l)} -P^{(k,l)}_t\right|}& \leq \sqrt{8} \frac{(k+l)^{\frac{1}2} \e^{D_1 \e^{D_2 t}/2}}{N_1^{1/2}}. 
\end{align}
The bounds \eqref{ineq:magmain0} and \eqref{ineq:srmain0} clearly imply the bounds \eqref{ineq:magmain} and \eqref{ineq:srmain} from Theorems \ref{thm:magnetic}, respectively \ref{thm:semirel}, in the case $\theta =0$. 
 
	Therefore, for $\theta=0$, it is enough to show \eqref{Gron1}.
	Using \eqref{an1} and \eqref{totalpde}, we obtain
	\begin{equation}
	\partial_t a_N^{(1)}(\Psi_{N,t}, \psi_t)=  \scp{ \Psi_{N,t}}{\i [ H_N, q_{1,1}^{\psi_t}] \Psi_{N,t} \rangle +  \langle \Psi_{N,t},  (\partial_t q_{1,1}^{\psi_t}) \Psi_{N,t}}.
	\end{equation}
	To simplify the first summand, we observe that since $q_{1,1}^{\psi_t}$ acts only on the first particle coordinate of the first species, the commutator with terms of $H_{N}$ that do not act on this coordinate
	vanishes. Moreover, the expression is symmetric in the coordinates $x_2^{(1)},\ldots, x_{N_1}^{(1)}$ and $x_1^{(2)},\ldots, x_{N_2}^{(2)}$ respectively. Using these observations and \eqref{hartree}, we arrive at 
	\begin{equation*}
	\begin{split}
	\partial_t a_{N}^{(1)}(\Psi_{N,t}, \psi_t)= &\i \scp{ \Psi_{N,t}}{ [ S_1^{(1)} + u^{(1,1)}(x_2^{(1)}-x_1^{(1)}) +  \frac1R u^{(1,2)}(x_1^{(2)}-x_1^{(1)}), q_{1,1}^{\psi_t}] \Psi_{N,t}} \\
	&-\i  \scp{\Psi_{N,t}}{  [S_1^{(1)} + u^{(1,1)} *|\psi_t|^2(x_1^{(1)}) + \frac1R u^{(1,2)} *|\vphi_t|^2(x_1^{(1)})  ,q_{1,1}^{\psi_t}] \Psi_{N,t}}.
	\end{split}	
	\end{equation*}
	After simplifying and using $p_{1,1}^{\psi_t}+q_{1,1}^{\psi_t}=1$ and the fact that for self-adjoint operators $C, D$ we have  $\i \langle \Psi_{N,t},  [C, D]  \Psi_{N,t} \rangle= -2\ima\langle \Psi_{N,t}, C D \Psi_{N,t} \rangle  $, we extract
	\begin{equation}\label{gronwall1}
	\partial_t a_{N}^{(1)}(\Psi_{N,t}, \psi_t)= -2 \ima(I+\frac1RII),
	\end{equation}
	where 
	\begin{align*}
	I\;&:=\;\scp{\Psi_{N,t}}{ p_{1,1}^{\psi_t} (u^{(1,1)}(x_2^{(1)}-x_1^{(1)})- u^{(1,1)}*|\psi_t|^2(x_1^{(1)})) q_{1,1}^{\psi_t} \Psi_{N,t} },\\
	II\;&:=\;\scp{ \Psi_{N,t}}{p_{1,1}^{\psi_t} (u^{(1,2)}(x_1^{(2)}-x_1^{(1)})- u^{(1,2)}*|\vphi_t|^2(x_1^{(1)}))  q_{1,1}^{\psi_t} \Psi_{N,t}}.
	\end{align*}
	We first estimate $II$.  Using $q_{1,2}^{\vphi_t} + p_{1,2}^{\vphi_t}=1$ twice, we split $II$ into
	\begin{equation}\label{bdec}
	II=II_1+II_2+II_3,
	\end{equation}
	with
	\begin{align*}
	II_1\;&:=\;\scp{\Psi_{N,t}}{   p_{1,1}^{\psi_t} p_{1,2}^{\vphi_t}(u^{(1,2)}(x_1^{(2)}-x_1^{(1)})- u^{(1,2)}*|\vphi_t|^2(x_1^{(1)}))  q_{1,1}^{\psi_t} p_{1,2}^{\vphi_t} \Psi_{N,t}},\\
	II_2\;&:=\; \scp{\Psi_{N,t}}{  p_{1,1}^{\psi_t} q_{1,2}^{\vphi_t} (u^{(1,2)}(x_1^{(2)}-x_1^{(1)})- u^{(1,2)}*|\vphi_t|^2(x_1^{(1)})) q_{1,1}^{\psi_t} p_{1,2}^{\vphi_t} \Psi_{N,t}},\\
	II_3\;&:=\;\scp{\Psi_{N,t}}{  p_{1,1}^{\psi_t} (u^{(1,2)}(x_1^{(2)}-x_1^{(1)})-u^{(1,2)}*|\vphi_t|^2(x_1^{(1)})) q_{1,1}^{\psi_t}  q_{1,2}^{\vphi_t} \Psi_{N,t}}.
	\end{align*}
	The equality $p_{1,2}^{\vphi_t} u^{(1,2)}(x_1^{(2)}-x_1^{(1)}) p_{1,2}^{\vphi_t}= p_{1,2}^{\vphi_t} u^{(1,2)}*|\vphi_t|^2(x_1^{(1)})
	p_{1,2}^{\vphi_t}$ immediately yields
	\begin{equation}\label{eq:b1van}
	II_1\;=\;0.
	\end{equation}		
	Due to $q_{1,2}^{\vphi_t} p_{1,2}^{\vphi_t}=0$, $II_2$ may be estimated by
	\begin{equation*}
	|II_2|\;\leqslant\;\|q_{1,2}^{\vphi_t}\Psi_{N,t}\|\|p_{1,1}^{\psi_t} u^{(1,2)}(x_1^{(2)}-x_1^{(1)})p_{1,2}^{\vphi_t} q_{1,1}^{\psi_t}\Psi_{N,t}\|.
	\end{equation*}
	Since $u^{(i,j)}$ is a Yukawa potential, see \eqref{eq:yukawa}, Hardy's inequality ($\frac1{|\cdot|^2}\leq-4\Delta$) and the diamagnetic inequality ($|\nabla|\Phi||\leq|D_A\Phi|$) imply that $u^{(i,j)}$ is form-bounded with respect to the magnetic Laplacian with form-bound zero. We even have
	\begin{equation}\label{eq:uij2}
	|u^{(i,j)}(\cdot-y)|^2\;\leq\;4\big(\lambda^{(i,j)}\big)^2D_A^2
	\end{equation}
	in the sense of quadratic forms. Henceforth, we obtain the estimates for the semi-relativistic case if we set $A=0$. Using Cauchy-Schwarz together with \eqref{eq:uij2}, \eqref{an1}, \eqref{an2} and the arithmetic mean-geometric mean inequality we arrive at
	\begin{equation}\label{b2est}
	|II_2| \leq |\lambda^{(1,2)}| \min(\|D_A\psi_t\|_2,\|D_A\vphi_t\|_2) (a_{N}^{(1)}(\Psi_{N,t},\psi_t)+a_{N}^{(2)}(\Psi_{N,t}, \vphi_t)).
	\end{equation}
	To control $II_3$ we argue as in \cite{Pe}. If we define $K_m:=(u^{(1,2)}(x_m^{(2)}-x_1^{(1)})-u^{(1,2)}*|\vphi_t|^2(x_1^{(1)})) p_{1,1}^{\psi_t}$ for every $m\in\{1;2;\ldots;N_2\}$, we find
	$$II_3=\scp{q_{1,2}^{\vphi_t} K_1\Psi_{N,t}}{q_{1,1}^{\psi_t} \Psi_{N,t}},$$
	%$B_3^{(2)}$ can be estimated by
	%\begin{equation}\label{eq:b32est}
	%\begin{split}
	%B _3^{(2)}\;&\leqslant\;
	%\|q_{N_1+semi1}^\vphi\Psi_N\|\|u^{(1,2)}*|\vphi|^2\|_{\infty} \|q_1^\psi\Psi_N\|\\
	%&\leqslant\;\frac{\pi}{4}\|\nabla\vphi\|_2^2(a_{N}^{(1)}(\Psi_N,\psi)+a_{N}^{(2)}(\Psi_N, \vphi)),
	%\end{split}
	%\end{equation}
	%where we used Kato's Inequality $\frac1{|\cdot|}\leqslant\frac{\pi}{2}(-\Delta)^{\frac12}$ again together with \eqref{an1} and \eqref{an2}.
	which then due to the symmetry of $\Psi_{N,t}$ gives
	$$N_2 II_3= \left\langle \sum_{m=1}^{N_2}   q_{m,2}^{\vphi_t} K_m \Psi_{N,t},  q_{1,1}^{\psi_t} \Psi_{N,t} \right\rangle.$$
	Therefore, by the Cauchy-Schwarz inequality it follows
	$$N_2 |II_3| \leq \|q_{1,1}^{\psi_t} \Psi_{N,t}\| \sqrt{\left\langle \sum_{m=1}^{N_2}   q_{m,2}^{\vphi_t} K_{m} \Psi_{N,t}, \sum_{n=1}^{N_2}   q_{n,2}^{\vphi_t} K_{n} \Psi_{N,t} \right\rangle}.$$

	Splitting the terms to the ones with $n=m$ and the ones with $n \neq m$, we obtain 
	$$|II_3| \leq \frac{\|q_{1,1}^{\psi_t} \Psi_{N,t}\|}{N_2} \sqrt{\sum_{m=1}^{N_2}  \left\langle  K_{m} \Psi_{N,t},   q_{m,2}^{\vphi_t} K_{m} \Psi_{N,t} \right\rangle +\sum_{m,n=1, m \neq n}^{N_2} \left\langle   K_{m} q_{n,2}^{\vphi_t} \Psi_{N,t},     K_{n} q_{m,2}^{\vphi_t} \Psi_{N,t} \right\rangle}  ,$$
	where in the last step we have used that $K_m$ commutes with $q_{n,2}^{\vphi_t}$ for $m\neq n$.
	Therefore, using \eqref{an1}, \eqref{an2} together with the symmetry of $\Psi_{N,t}$, we find 
	$$|II_3| \leq \sqrt{a_{N}^{(1)}(\Psi_{N,t}, \psi_t)} \sqrt{\frac{\| K_1\|^2}{N_2}  +\|K_1\|^2 a_{N}^{(2)}(\Psi_{N,t}, \vphi_t)  }  .$$
	We have
	\begin{align*}
	\|K_1\|\;&\leq\;\|u^{(1,2)}(x_1^{(2)}-x_1^{(1)})p_{1,1}^{\psi_t}\|+\|u^{(1,2)}*|\vphi_t|^2\|_\infty\\
	&\leq  2|\lambda^{(1,2)}|\left( \|D_A\psi_t\|_2+\|D_A\vphi_t\|_2\right),
	\end{align*}
	where we again used Cauchy-Schwarz and \eqref{eq:uij2}. Note that the calculation also uses that $\|\vphi_t\|_{2}=1$ as follows from \ref{itm:rm}, \ref{itm:sr} and the conservation of the $L^2$-norm of $\vphi_t, \psi_t$, see Appendix \ref{app:mag}. It follows
	$$ |II_3| \leq  |\lambda^{(1,2)}|\left( \|D_A\psi_t\|_2+\|D_A\vphi_t\|_2\right) \left(\frac{1}{N_2} + a_{N}^{(1)}(\Psi_{N,t}, \psi_t)+a_{N}^{(2)}(\Psi_{N,t}, \vphi_t)\right), $$
	where we have also used the arithmetic mean-geometric mean inequality. As a consequence of the last estimate together with \eqref{bdec}, \eqref{eq:b1van} and \eqref{b2est}, we find
	\begin{equation}\label{btotest}
	\begin{split}
	|II|&\leq |II_2|+|II_3|\\
	&\leq   2 |\lambda^{(1,2)}|\left( \|D_A\psi_t\|_2+\|D_A\vphi_t\|_2\right)\left(\frac{1}{N_2} + a_{N}^{(1)}(\Psi_{N,t}, \psi_t)+a_{N}^{(2)}(\Psi_{N,t}, \vphi_t)\right).
	\end{split}
	\end{equation}
	Splitting $I$ analogously to $II$ by using $p_2^{\psi_t}+q_2^{\psi_t}=1$, we obtain
	\begin{equation}\label{est:TeilI}
	|I|\;\leq\; 2 |\lambda^{(1,1)}| \|D_A\psi_t\|_2 \left(\frac{1}{N_1-1} + 2 a_{N}^{(1)}(\Psi_{N,t}, \psi_t)\right),
	\end{equation}
	where 
	$$\ima{\scp{\Psi_{N,t}}{  p_{1,1}^{\psi_t} q_{2,1}^{\psi_t} (u^{(1,1)}(x_2^{(1)}-x_1^{(1)})- u^{(1,1)}*|\psi_t|^2(x_1^{(1)})) q_{1,1}^{\psi_t} p_{2,1}^{\psi_t} \Psi_{N,t}}}$$$$=\ima{\scp{\Psi_{N,t}}{  p_{1,1}^{\psi_t} q_{2,1}^{\psi_t} u^{(1,1)}(x_2^{(1)}-x_1^{(1)}) q_{1,1}^{\psi_t} p_{2,1}^{\psi_t} \Psi_{N,t}}}$$ 
	vanishes due to the symmetry in $(x_1^{(1)},p_{1,1}^{\psi_t})\leftrightarrow (x_2^{(1)},p_{1,2}^{\psi_t})$. Now estimate \eqref{est:TeilI} together with \eqref{btotest} inserted in \eqref{gronwall1} yields
	\begin{equation}\label{est:dera1}
	\partial_t a_N^{(1)}(\Psi_{N,t},\psi_t)\leq \frac{4 |\lambda^{(1,2)}|}{R} \left(\|D_A\psi_t\|_2+\|D_A\vphi_t\|_2\right)\left(\frac{1}{N_2} + a_{N}^{(1)}(\Psi_{N,t}, \psi_t)+a_{N}^{(2)}(\Psi_{N,t}, \vphi_t)\right)
$$$$	+4|\lambda^{(1,1)}| \|D_A\psi_t\|_2 \left(\frac{1}{N_1-1} + 2 a_{N}^{(2)}(\Psi_{N,t}, \vphi_t)\right).
	\end{equation}
Similar calculations give  	
		\begin{equation}\label{est:dera2}
		\partial_t a_N^{(2)}(\Psi_{N,t},\vphi_t)\leq 4 |\lambda^{(1,2)}| R\left(\|D_A\psi_t\|_2+\|D_A\vphi_t\|_2\right)\left(\frac{1}{N_1} + a_{N}^{(1)}(\Psi_{N,t}, \psi_t)+a_{N}^{(2)}(\Psi_{N,t}, \vphi_t)\right)
		$$$$	+4 |\lambda^{(2,2)}| \|D_A\vphi_t\|_2 \left(\frac{1}{N_2-1} + 2 a_{N}^{(1)}(\Psi_{N,t}, \psi_t)\right).
		\end{equation}
Estimate	\eqref{eq:derunbdd} in appendix \ref{app:mag} shows that in the magnetic case there exists $\tilde{C} > 0$ such that for all $t$ we have
	\begin{equation}\label{est:nablamag}
\|D_A\psi_t\|_2,\|D_A\vphi_t\|_2 \leq \tilde{C}.
	\end{equation}
As  explained at the end of appendix \ref{app:sr} in the semi-relativistic case there exist $\tilde{C},M>0$ such that 
	\begin{equation}\label{est:nablasr}
	\|\nabla \psi_t\|_2,\|\nabla \vphi_t\|_2 \leq \tilde{C} \e^{Mt}.
	\end{equation}
Using \eqref{est:dera1}, \eqref{est:dera2}, \eqref{est:nablamag} and \eqref{est:nablasr}
we arrive at  \eqref{Gron1}, as desired. This finishes the proof of Theorems \ref{thm:magnetic} and \ref{thm:semirel} in the case $\theta=0$. 

\subsection{The case $\theta \in (0,1)$.}\label{sec:theta-0-1}
To show \eqref{ineq:magmain} and \eqref{ineq:srmain} for  $\theta \in (0,1)$ we will use Theorem \ref{thm:int} below which is
 the two-species analog  of Theorem 2.1 in \cite{AH}. %It can be proven in a very similar way to Theorem 1.1 in \cite{AH} if one also takes Remark 1.2 thereafter into account. 
\par Let $\Psi \in \left(L^2(\R^{3})^{\otimes_S N_1} \otimes L^2(\R^{3})^{\otimes_S N_2}\right) \cap H^1(\mathbb{R}^{3(N_1+N_2)})$, $\psi,\vphi \in H^1(\mathbb{R}^{3})$ with $\|\Psi\|_{L^2}=\|\psi\|_{L^2}=\|\vphi\|_{L^2}=1$. For $k \in \{1, \dots, N_1\}$, $l \in \{1, \dots, N_2\}$ we define $$\gamma_N^{(k,l)}=\tr_{k+1,...,N_1+N_2-l}(\ket{\Psi}\bra{\Psi}),$$ $P^{(k,l)}:=|\psi^{\otimes k} \rangle \langle \psi^{\otimes k}| \otimes |\vphi^{\otimes l} \rangle \langle \vphi^{\otimes l}|$, and let $S_{k,l,\theta}$ be defined as in \eqref{eq:s-k-l-theta}. Similarly to the proof of \eqref{ineq:magmain} and \eqref{ineq:srmain} for  $\theta=0$  above, we introduce
\begin{equation}
\begin{split}
a_{N,1}\;&:=\;\langle \Psi, q_{1,1}^\vphi \Psi \rangle = \|q_{1,1}^\vphi \Psi\|^2,\\ a_{N,2}\;&:=\;\langle \Psi, q_{1,2}^\psi \Psi \rangle = \|q_{1,2}^\psi \Psi\|^2,
\end{split}
\end{equation}
where $q_{i,j}^h$ is defined in \eqref{eq:qij}, and we denote the Hilbert-Schmidt norm of an operator acting on $L^2$ by $\|.\|_{HS}$. 

\begin{theorem}\label{thm:int}
	For any $\theta \in (0,1)$ we have the estimate
	\begin{equation*}
	\tr\left|S_{k,l,\theta }^{\frac12} (\gamma_N^{(k,l)}-P^{(k,l)}) S_{k,l,\theta }^{\frac12}\right|\leq C(k+l)\left( a_{N,1}^{\min(\frac{1}{2},1-\theta)}+a_{N,2}^{\min(\frac{1}{2},1-\theta)}+\|\gamma_N^{(k,l)}-P^{(k,l)}\|_{HS}^{1-\theta}\right),
	\end{equation*}
where $C=C_{\Psi,\psi,\vphi}$ is defined as 
\begin{align*}
C_{\Psi,\psi,\vphi}:=2\Big(\|S_{1,1,\frac{1}{2}}\Psi\|_2+\|\big(1+S^{(1)}\big)^{1/2}\psi\|_2+\|\big(1+S^{(2)}\big)^{1/2}\vphi\|_2\Big)^{2}
\end{align*}
\end{theorem}
\begin{proof}
	The proof is a relatively straightforward adaptation of the proof of Theorem 2.1 in \cite{AH}. For convenience of the reader we will outline the argument. We start with the estimate
	\begin{equation}\label{est:HSHtheta}
	\begin{split}
	\tr\left|S_{k,l,\theta}^{\frac12}(\gamma_{N}^{(k,l)}-P^{(k,l)})S_{k,l,\theta}^\frac12\right|
	\;\leq\; &2 \left\|S_{k,l,\theta}^\frac12(\gamma_{N}^{(k,l)}-P^{(k,l)})S_{k,l,\theta}^\frac12\right\|_{HS}\\
	&+\tr\left(S_{k,l,\theta}^\frac12(\gamma_{N}^{(k,l)}-P^{(k,l)})S_{k,l,\theta}^\frac12\right).
	\end{split}
	\end{equation} 
	This can be proven similarly as the first bound in \eqref{aprest}, using that 
	$S_{k,l,\theta}^\frac12P^{(k,l)} S_{k,l,\theta}^\frac12$ is again a rank one operator, but now with the difference that  $\tr\big(S_{k,l,\theta}^\frac12(\gamma_{N}^{(k,l)}-P^{(k,l)})S_{k,l,\theta}^\frac12\big)$ does not necessarily  vanish. 
	Next proceeding as in \cite{AH} we can show that 
	\begin{equation}\label{est:Hoelder}
	\left\|S_{k,l,\theta}^\frac12(\gamma_{N}^{(k,l)}-P^{(k,l)})S_{k,l,\theta}^\frac12\right\|_{HS}
	\;\leq\; \left\|S_{k,l,1}^\frac12(\gamma_{N}^{(k,l)}-P^{(k,l)})S_{k,l,1}^\frac12\right\|_{HS}^\theta
	\left\|\gamma_{N}^{(k,l)}-P^{(k,l)}\right\|_{HS}^{1-\theta}.
	\end{equation}
	The only difference is that \cite{AH} works with the Fourier transform, which we cannot use anymore in the magnetic case.
		 Instead of using the Fourier transform, one can work with the spectral theorem, which represents $D_A^2$ as a multiplication operator in an appropriate $L^2$-space, see \cite[Theorem VIII.4]{RS1}. 
	Other than that the argument uses as in \cite{AH} the fact that the Hilbert-Schmidt norm of an operator is the $L^2$-norm of its Kernel together with Hölder's inequality. Note also that using the fact that the Hilbert-Schmidt norm of an operator is less or equal than its trace norm, together with the triangle inequality for the trace norm, we can arrive at 
	\begin{equation*}
	\left\|S_{k,l,1}^\frac12(\gamma_{N}^{(k,l)}-P^{(k,l)})S_{k,l,1}^\frac12\right\|_{HS} \leq \tr\left(S_{k,l,1}^\frac12 \gamma_{N}^{(k,l)} S_{k,l,1}^\frac12\right)+\tr\left(S_{k,l,1}^\frac12P^{(k,l)} S_{k,l,1}^\frac12\right),
	\end{equation*}
	since the operators on the right hand side are positive. We now observe that
	\begin{equation*}
	 \tr\left(S_{k,l,1}^\frac12P^{(k,l)} S_{k,l,1}^\frac12\right)=\langle\psi^{\otimes k} \otimes \vphi^{\otimes l}, S_{k,l,1} \psi^{\otimes k} \otimes \vphi^{\otimes l} \rangle
	 \end{equation*}
	  and from the definition of the partial trace it follows that \begin{equation*}\tr\left(S_{k,l,1}^\frac12 \gamma_{N}^{(k,l)} S_{k,l,1}^\frac12\right)=\langle \Psi, S_{k,l,1} \Psi \rangle.
	  \end{equation*}
	Therefore, we can easily obtain that 
		\begin{equation*}
		\left\|S_{k,l,1}^\frac12(\gamma_{N}^{(k,l)}-P^{(k,l)})S_{k,l,1}^\frac12\right\|_{HS} \leq \langle \Psi, S_{k,l,1} \Psi \rangle + \langle\psi^{\otimes k} \otimes \vphi^{\otimes l}, S_{k,l,1} \psi^{\otimes k} \otimes \vphi^{\otimes l} \rangle,
		\end{equation*}
		which together with \eqref{est:Hoelder} implies that
		\begin{equation}\label{est:HSHtheta1}
			\left\|S_{k,l,\theta}^\frac12(\gamma_{N}^{(k,l)}-P^{(k,l)})S_{k,l,\theta}^\frac12\right\|_{HS}$$$$
			\;\leq\; \left( \langle \Psi, S_{k,l,1} \Psi \rangle + k \langle \psi, (1+S^{(1)}) \psi \rangle+l \langle \vphi,  (1+S^{(2)})  \vphi \rangle  \right)^\theta
			\left\|\gamma_{N}^{(k,l)}-P^{(k,l)}\right\|_{HS}^{1-\theta}.
			\end{equation}
	A straightforward rewritting of the arguments in \cite{AH} gives 
	\begin{equation}\label{est:HSHtheta2}
	\tr\left(S_{k,l,\theta}^\frac12(\gamma_{N}^{(k,l)}-P^{(k,l)})S_{k,l,\theta}^\frac12\right) $$$$\leq  k (\|S_1^{(1)} \Psi\|_2+\|S^{(1)}\psi\|_2)^{\max(1,2\theta)} a_{N,1}^{\min(\frac{1}{2},1-\theta)}+l (\|S_1^{(2)} \Psi\|_2+\|S^{(2)}\vphi\|_2)^{\max(1,2\theta)} a_{N,2}^{\min(\frac{1}{2},1-\theta)}.
	\end{equation}
	Using \eqref{est:HSHtheta}, \eqref{est:HSHtheta1} and \eqref{est:HSHtheta2}, one can with  elemantary computations finish the proof of the theorem.
\end{proof}

We now show how to use Theorem \ref{thm:int} in order to prove \eqref{ineq:magmain} and \eqref{ineq:srmain} for $0<\theta<1$: In appendix \ref{app:self-adjoint} we show that there exists a constant $C>0$ such that 
\begin{equation}\label{ineq:averagekinetic}
\|S_{1,1,\frac{1}{2}} \Psi_{N,t}\| \leq C,
\end{equation}
where in the semi-relativistic case assumption \ref{itm:sr} is used.
The bounds \eqref{ineq:magmain} and \eqref{ineq:srmain} for $\theta \in (0,1)$ then  follow from Theorem \ref{thm:int} together with \eqref{est:fmag}, \eqref{est:fsr}, \eqref{est:aN}, \eqref{ineq:magmain0}, \eqref{ineq:srmain0},
\eqref{est:nablamag}, \eqref{est:nablasr}, \eqref{ineq:averagekinetic} and the standard inequality $\|A\|_{HS} \leq \tr|A|$. This finishes the proof of Theorems \ref{thm:magnetic} and \ref{thm:semirel} in the case $\theta \in (0,1)$. It remains to prove \eqref{conv:magen}, i.e. Theorem \ref{thm:magnetic} in the case $\theta=1$, which we do in the next section.

\subsection{Proof of \eqref{conv:magen}}
We will follow the idea of the proof of part (ii) of Theorem 1.1 in \cite{AH}. However, in the case of two (or more) different species of Bosons there are significant additional difficulties. One of the main reasons for these difficulties is the fact that in \cite{AH}, it was enough to prove convergence of the potential energy, which, together with energy conservation, also immediately implied the convergence of the kinetic energy. This was then used to show convergence in the energy trace norm. In the case of two or more different species, it is not a-priori clear how the kinetic energy distributes between the different species even if one knows that the potential energies converge. 
This leads to significant modifications in the argument.   
\par Let $t>0$ and $k,l\in\N$. We start first with 

\paragraph{Case $k=R^2 l$.}
This is the easier case, since then we have $N_1/N_2= k/l$. We start as in \cite{AH} with the estimate
\begin{equation}\label{est:HSH1}
\begin{split}
\tr\left|S_{k,l,1}^{\frac12}(\gamma_{N,t}^{(k,l)}-P^{(k,l)}_t)S_{k,l,1}^\frac12\right|
\;\leq\; &2 \left\|S_{k,l,1}^\frac12(\gamma_{N,t}^{(k,l)}-P^{(k,l)}_t)S_{k,l,1}^\frac12\right\|_{HS}\\
&+\tr\left(S_{k,l,1}^\frac12(\gamma_{N,t}^{(k,l)}-P^{(k,l)}_t)S_{k,l,1}^\frac12\right).
\end{split}
\end{equation} 
which can be proven in the same way as \eqref{est:HSHtheta}. We next show that
\begin{equation}\label{trgegen0}
\lim_{\substack{N_i\to\infty\\ N_1/N_2=R^2}}\tr\left(S_{k,l,1}^{1/2}(\gamma_{N,t}^{(k,l)}-P^{(k,l)}_t)S_{k,l,1}^{1/2}\right)\;=\;0.
\end{equation}
To prove \eqref{trgegen0}, we use the symmetry of $\Psi_{N,t}$ to obtain
\begin{equation}\label{dec:tr}
\begin{split}
\tr\left(S_{k,l,1}^\frac12(\gamma_{N,t}^{(k,l)}-P^{(k,l)}_t)S_{k,l,1}^\frac12\right)\;=\;k \left(\scp{\Psi_{N,t}}{ S_1^{(1)} \Psi_{N,t}}-\scp{\psi_{t}}{ S^{(1)} \psi_{t}}\right)\\
+l\left(\scp{\Psi_{N,t}}{ S_1^{(2)} \Psi_{N,t}}-\scp{\vphi_{t}}{ S^{(2)} \vphi_{t}}\right).
\end{split}
\end{equation}
The energy conservation, see Appendix \ref{app:wellposed}, and the initial condition $\Psi_{N,0}=\psi_0^{\otimes N_1}\otimes\vphi_0^{\otimes N_2}$ yield 
\begin{equation}\label{eq:econs}
\frac{1}{N_2}\langle \Psi_{N,t}, H_N \Psi_{N,t} \rangle=\frac1{N_2}\cH_N(\overline{\psi_0},\overline{\vphi_0},\psi_0,\vphi_0)=\frac1{N_2}\cH_N(\overline{\psi_t},\overline{\vphi_t},\psi_t,\vphi_t),
\end{equation}
where $\cH_N$ was defined in \eqref{eq:chndef}, and therefore the symmetry with respect to the particle coordinates for each species individually gives
\begin{equation}\label{tracepart}
\begin{split}
\scp{\Psi_{N,t}}{ S_1^{(2)} \Psi_{N,t}}-\scp{\vphi_t}{ S^{(2)} \vphi_t}&+R^2\left(\scp{\Psi_{N,t}}{ S_1^{(1)} \Psi_{N,t}}-\scp{\psi_t}{ S^{(1)} \psi_t}\right)\\=-\frac{R^2}{2}&\left(\scp{ \Psi_{N,t}}{ u^{(1,1)}_{1,2} \Psi_{N,t}}- \scp{\psi_t^{\otimes 2}}{ u^{(1,1)}_{1,2} \psi_t^{\otimes 2}}\right)\\
-\frac{1}{2}&\left(\scp{ \Psi_{N,t}}{ u^{(2,2)}_{1,2} \Psi_{N,t} }- \scp{\vphi_t^{\otimes 2}}{ u^{(2,2)}_{1,2} \vphi_t^{\otimes 2}}\right)\\
-R&\left(\scp{ \Psi_{N,t}}{ u^{(1,2)}_{1,1} \Psi_{N,t} }- \scp{\psi_t\otimes\vphi_t}{ u^{(1,2)}_{1,1}\psi_t\otimes\vphi_t}\right).
\end{split}
\end{equation}
We have used the abbreviation $u^{(i,j)}_{k,l}\;:=\;u^{(i,j)}(x_k^{(i)}-x_l^{(j)})$. We note that the square root is operator monotone, meaning that $0 \leq A \leq B \implies \sqrt{A} \leq \sqrt{B}$, see for example Theorem 2.6 in \cite{SU}. Therefore, from \eqref{eq:uij2}, we deduce that the operators 
\begin{align*} 
S_{2,0,1}^{-1/4} u^{(1,1)}_{m,n} S_{2,0,1}^{-1/4},\,  S_{0,2,1}^{-1/4} u^{(2,2)}_{m,n} S_{0,2,1}^{-1/4}, \text{ for } m\not=n, \text{ and }  
	S_{1,1,1}^{-1/4} u^{(1,2)}_{m,n} S_{1,1,1}^{-1/4}, \text{ for all } m,n, 
\end{align*}
where $S_{k,l,1}$ was defined in \eqref{eq:s-k-l-theta}, 
can be extended to bounded operators on $L^2(\R^3)^{\otimes_S2}$. Now  we can use 
 Corollary \ref{cor:observ}, whose proof only used  the bound \eqref{ineq:magmain} of Theorem \ref{thm:magnetic}, which we already proved in Sections \ref{sec:theta-0} and \ref{sec:theta-0-1}. Applying Corollary \ref{cor:observ} to $\cA=u^{(i,j)}_{k,l}$ we obtain 
\begin{align*}
\lim_{\substack{N_i\to\infty\\ N_1/N_2=R^2}}	\scp{ \Psi_{N,t}}{ u^{(1,1)}_{1,2} \Psi_{N,t}}\;&=\; \scp{\psi_t^{\otimes 2}}{ u^{(1,1)}_{1,2} \psi_t^{\otimes 2}},\\
\lim_{\substack{N_i\to\infty\\ N_1/N_2=R^2}}	\scp{ \Psi_{N,t}}{ u^{(2,2)}_{1,2} \Psi_{N,t}}\;&=\; \scp{\vphi_t^{\otimes 2}}{ u^{(2,2)}_{1,2} \vphi_t^{\otimes 2}},\\
\lim_{\substack{N_i\to\infty\\ N_1/N_2=R^2}}	\scp{ \Psi_{N,t}}{ u^{(1,2)}_{1,1} \Psi_{N,t}}\;&=\; \scp{\psi_t\otimes \vphi_t}{ u^{(1,2)}_{1,1} \psi_t\otimes\vphi_t},
\end{align*}
which together with \eqref{dec:tr}, \eqref{tracepart} and the assumption $k=R^2 l$ gives \eqref{trgegen0}. Now 
\begin{equation}\label{HSfirstcase}
\|S_{k,l,1}^{1/2}(\gamma_{N,t}^{(k,l)}-P^{(k,l)}_t)S_{k,l,1}^{1/2}\|_{HS}\to 0 \text{ as }N_1,N_2\to\infty, N_1/N_2=R^2
\end{equation}  can be analogously proven as in \cite{AH}. Using \eqref{est:HSH1}, \eqref{trgegen0} and \eqref{HSfirstcase}, we arrive at  \eqref{conv:magen} in the case $k=R^2 l$. 
\paragraph{Case $k\neq R^2l$.} Let $k_0\geq k$, $l_0\geq l$ be such that $k_0=R^2l_0$. Using that the Hilbert-Schmidt-norm of an operator is the $L^2$-norm of its integral kernel we find
\begin{equation}\label{eq:hsstronger}
\left\|S_{k,l,1}^\frac12\left(\gamma_{N,t}^{(k_0,l_0)}-P^{(k_0,l_0)}_t\right)S_{k,l,1}^\frac12\right\|_{HS}\;\leq\;\left\|S_{k_0,l_0,1}^\frac12\left(\gamma_{N,t}^{(k_0,l_0)}-P^{(k_0,l_0)}_t\right)S_{k_0,l_0,1}^\frac12\right\|_{HS}\rightarrow0,
\end{equation}
as $N_1,N_2\to\infty$, $N_1/N_2\to\infty$ due to \eqref{HSfirstcase}. In addition, we have
\begin{equation}\label{partialtr}
S_{k,l,1}^\frac12\left(\gamma_{N,t}^{(k,l)}-P^{(k,l)}_t\right)S_{k,l,1}^\frac12\;=\;\tr_{k+1,\ldots,k_0+l_0-l}{\left(S_{k,l,1}^\frac12\left(\gamma_{N,t}^{(k_0,l_0)}-P^{(k_0,l_0)}_t\right)S_{k,l,1}^\frac12\right)},
\end{equation}
where we have traced out with respect to the coordinates that $S_{k,l,1}$ does not act on. Define $A_N:=\scp{\Psi_{N_t}}{ S_1^{(1)} \Psi_{N,t}}-\scp{\psi_{t}}{ S^{(1)} \psi_{t}}$ and $B_N:=\scp{\Psi_{N_t}}{ S_1^{(2)} \Psi_{N,t}}-\scp{\vphi_{t}}{ S^{(2)} \vphi_{t}}$. Similarly to \eqref{dec:tr}, we then obtain
\begin{equation}\label{eq:kanlbn}
\tr{\left(S_{k,l,1}^\frac12\left(\gamma_{N,t}^{(k_0,l_0)}-P^{(k_0,l_0)}_t\right)S_{k,l,1}^\frac12\right)}\;=\;k A_N+l B_N.
\end{equation}
Using this together with 
\begin{equation}\label{eq:skk0}
\begin{split}
\tr{\left|S_{k,l,1}^\frac12\left(\gamma_{N,t}^{(k_0,l_0)}-P^{(k_0,l_0)}_t\right)S_{k,l,1}^\frac12\right|}\leq &\tr{\left(S_{k,l,1}^\frac12\left(\gamma_{N,t}^{(k_0,l_0)}-P^{(k_0,l_0)}_t\right)S_{k,l,1}^\frac12\right)}\\
&+2\left\|S_{k,l,1}^\frac12\left(\gamma_{N,t}^{(k_0,l_0)}-P^{(k_0,l_0)}_t\right)S_{k,l,1}^\frac12\right\|_{HS},
\end{split}
\end{equation}
and \eqref{eq:hsstronger}, we find 
\begin{displaymath} 
\liminf_{N_i\to\infty,N_1/N_2=R^2}(kA_N+lB_N)\;\geq \;0.	
\end{displaymath}
Since this holds for arbitrary values of $k\leq N_1$, $l\leq N_2$, we may use $k=0$, respectively $l=0$,  to see that
\begin{displaymath}
\liminf_{N_i\to\infty,N_1/N_2=R^2}A_N\geq0\quad\mbox{and }\quad\liminf_{N_i\to\infty,N_1/N_2=R^2}B_N\geq0.	
\end{displaymath}
From the already discussed case $k=R^2 l$ we find $\lim_{N_i\to\infty,N_1/N_2=R^2}(R^2A_N+B_N)=0$ which then yields $\lim_{N_i\to\infty,N_1/N_2=R^2}A_N=0$ and $\lim_{N_i\to\infty,N_1/N_2=R^2}B_N=0$. 

Thus  \eqref{eq:kanlbn} implies $\lim_{N_i\to\infty,N_1/N_2=R^2}\tr{(S_{k,l,1}^{1/2}(\gamma_{N,t}^{(k_0,l_0)}-P^{(k_0,l_0)}_t)S_{k,l,1}^{1/2})}=0$. 
This together with \eqref{eq:hsstronger} and \eqref{eq:skk0} gives $$\tr{\left|S_{k,l,1}^\frac12\left(\gamma_{N,t}^{(k_0,l_0)}-P^{(k_0,l_0)}_t\right)S_{k,l,1}^\frac12\right|}\to0\quad\mbox{as\ }N_i\to\infty,N_1/N_2=R^2.$$
Using  Lemma \ref{lem:partialtracenorm} below and \eqref{partialtr}, this implies
$$\tr{\left|S_{k,l,1}^\frac12\left(\gamma_{N,t}^{(k,l)}-P^{(k,l)}_t\right)S_{k,l,1}^\frac12\right|}\;\leq\;\tr{\left|S_{k,l,1}^\frac12\left(\gamma_{N,t}^{(k_0,l_0)}-P^{(k_0,l_0)}_t\right)S_{k,l,1}^\frac12\right|}\to 0$$
as $N_i\to\infty,N_1/N_2=R^2$.
\begin{lemma}\label{lem:partialtracenorm}
Let $\rho$ be a trace class operator acting on $L^2(\R^3)^{\otimes m}$ and
$\tr_{q+1,\ldots,m}\rho$ be its partial trace over the last $m-q$ particle coordinates, where  $0<q \leq m$. Then $ \tr |\tr_{q+1,\ldots,m}\rho| \leq \tr|\rho| $.
\end{lemma}
\begin{proof}
	The proof is well-known but for convenience of the reader we will outline it. It is well-known, see for example \cite[Theorem VI.26]{RS1}, that the dual of the trace class operators acting on a Hilbert space $H$ is the space $B(H)$ of bounded operators acting on $H$. Therefore, 
\begin{equation*}
\begin{split}
\tr|\tr_{q+1,\ldots,m}\rho|\;&=\;\sup_{\|a\|_{B(L^2(\mathbb{R}^3)^{\otimes
			q})}=1}|\tr(a\tr_{q+1,\ldots,m}\rho)|\\&=\;\sup_{\|a\|_{B(L^2(\mathbb{R}^3)^{\otimes
			q})}=1}|\tr((a \otimes I^{\otimes m-q})
\rho)| \\&\;\leq \sup_{\|A\|_{B(L^2(\mathbb{R}^{3})^{\otimes m})}=1}|\tr (A\rho)|= \tr|\rho|.
\end{split}
\end{equation*}	
\end{proof}	
	
	\begin{remark}
From the proofs one sees that there are some generalizations of Theorems \ref{thm:magnetic} and \ref{thm:semirel}. A slight modification of the arguments shows that part (ii) of Theorem \ref{thm:magnetic}  remains true if  instead of factorized  initial data we assume
  \begin{align*}
  \lim_{\substack{N_1,N_2\to\infty\\ N_1/
  	N_2=R^2}}&\tr|\gamma_{N,0}^{(1,0)}-P^{(1,0)}_0|\,=\,\lim_{\substack{N_1,N_2\to\infty\\ N_1/N_2=R^2}}\tr|\gamma_{N,0}^{(0,1)}-P^{(0,1)}_0|\\
  &=\lim_{\substack{N_1,N_2\to\infty\\ N_1/N_2=R^2}}\frac{1}{N_1} \left(\langle \Psi_{N,0}, H_N \Psi_{N,0} \rangle-\langle \psi_0^{\otimes N_1} \otimes  \vphi_0^{\otimes N_2}, H_N \psi_0^{\otimes N_1} \otimes  \vphi_0^{\otimes N_2} \rangle \right)=0.
  \end{align*}
Part (i) of Theorem \ref{thm:magnetic} and Theorem \ref{thm:semirel} can also be generalized to cases where the initial data does not have to be factorized.  For example, if we assume \ref{itm:rm} and that $\langle \Psi_{N,0}, H_N \Psi_{N,0} \rangle$ is finite, then for any $R>0$ there are constants $C,D>0$ depending on $\frac{1}{N_1}\langle \Psi_{N,0}, H_N \Psi_{N,0} \rangle$ such that for all $t>0$, $0\leq k\leq N_1$, $0\leq l\leq N_2$ and $\theta\in[0,1)$ we have
	\begin{align}\nonumber
	\tr{\left|S_{k,l,\theta}^{\frac12}\left(\gamma_{N,t}^{(k,l)} -P^{(k,l)}_t\right)S_{k,l,\theta}^{\frac12}\right|}& \leq C(k+l)\left(\frac{(k+l)}{N_1}+a_{N}^{(1)}(\Psi_{N,0},\psi_0)+a_{N}^{(2)}(\Psi_{N,0},\vphi_0)\right)^{\frac{1-\theta}{2}}\e^{Dt}.
	\end{align}
 Moreover, in the magnetic case the interaction potentials do not have to be Coulomb or Yukawa potentials; an assumption of the form $|u^{(i,j)}|^2 \leq C (1-\Delta)$ for some $C>0$ would suffice. Furthermore, under certain regularity conditions so that \eqref{hartree} remains well-posed and the $H^1$- or $H_A^1$-norms of the involved fields do not blow up, one may add an external potential $V_{ext}$. Indeed, observe that the calculation of $\partial_ta_N^{(i)}$ in the proof remains unchanged if we replace $S_j^{(i)}$ by $S_j^{(i)}+V_{ext}(x_{j}^{(i)})$. 
		\end{remark}

\appendix 

\section{Proof of \eqref{ineq:averagekinetic} and of self-adjointness of $H_N$ in the semi-relativistic case }\label{app:self-adjoint}

Using Kato's inequality
\begin{equation} \label{ineq:Kato}
\frac{1}{|x|} \leq \frac{\pi}{2}(-\Delta)^{1/2},
\end{equation}
see \cite{Her}, we obtain that
\begin{equation*}
\sum_{j=1}^{N_i} \frac{1}{N_i-1} \sum_{k=j+1}^{N_i} u^{(i,i)}(x_{j}^{(i)}-x_{k}^{(i)}) \geq -\frac{\pi \lambda_-^{(i,i)}}{4} \sum_{j=1}^{N_i} \sqrt{m_i-\Delta_{x_j^{(i)}}}
\end{equation*}
and therefore, for all $\Psi_N \in H^{1/2}(\R^{3(N_1+N_2)})$.
\begin{equation*}
\bigg\langle \Psi_N, \sum_{j=1}^{N_i} \frac{1}{N_i-1} \sum_{k=j+1}^{N_i} u^{(i,i)}(x_{j}^{(i)}-x_{k}^{(i)}) \Psi_N \bigg\rangle \geq -\frac{\pi \lambda_-^{(i,i)}}{4} N_i a_i^2,
\end{equation*}
where $a_i:=\|(m_i-\Delta_{x_1^{(i)}})^\frac{1}{4} \Psi_N\|_2$. Since the interaction potential between the two different species can be estimated by the kinetic energies of any of the species, we see that 
\begin{equation*}
\bigg\langle \Psi_N, \sum_{i=1}^{N_1} \sum_{j=1}^{N_2} u^{(1,2)}(x_{i}^{(1)}-x_{j}^{(2)}) \Psi_N \bigg\rangle\geq -\frac{\pi \lambda_-^{(1,2)}}{2} N_1 N_2\min\left(a_1^2,a_2^2\right) \geq -\frac{\pi \lambda_-^{(1,2)}}{2} N_1 N_2 a_1 a_2.
\end{equation*}
Using the above inequalities, \ref{itm:mf}, and completion of the square, we arrive at
\begin{equation}\label{srsa}
\begin{split}
\langle \Psi_N, H_N \Psi_N \rangle \geq& \sum_{i=1}^2\left(1-\frac{\pi \lambda_-^{(i,i)}}{4}\right) N_i a_i^2 -\frac{\pi \lambda_-^{(1,2)}}{2} \sqrt{N_1 N_2} a_1 a_2\\
=&\frac{\pi}{4} \left(\sqrt{\frac{4}{\pi}- \lambda_-^{(1,1)}} \sqrt{N_1} a_1-\sqrt{\frac{4}{\pi}- \lambda_-^{(2,2)}} \sqrt{N_2} a_2\right)^2\\
&+\frac{\pi}{2} \sqrt{N_1 N_2} a_1 a_2 \left( \sqrt{\left(\frac{4}{\pi}- \lambda_-^{(1,1)}\right) \left(\frac{4}{\pi}- \lambda_-^{(2,2)}\right)}-\lambda_-^{(1,2)}\right).
\end{split}
\end{equation}
Using  now \ref{itm:sr} it follows that $H_N$ defines a positive quadratic form, which by Kato's inequality is closed in $H^{1/2}$. Therefore, $H_N$ can be realized as a self-adjoint operator see Theorem 2, Chapter 10 in \cite{BS}. 

The inequality \eqref{ineq:averagekinetic}
follows now from \eqref{srsa} together with the conservation of $\langle \Psi_{N,t}, H_N \Psi_{N,t} \rangle$ and the uniform boundedness of $\langle \Psi_{N,0}, H_N \Psi_{N,0} \rangle/\sqrt{N_1 N_2}$ with respect to $N_1, N_2$, when $\sqrt{N_1/N_2}=R$.

\section{Proof of Theorem \ref{thm:wellposed} }\label{app:wellposed}
\subsection{Magnetic case \label{app:mag}}
\paragraph{Local well-posedness.} A mild solution of the initial value problem \eqref{hartree} can be written as the system 
\begin{equation}\label{eq:integral}
\begin{split}
\left(\begin{array}{cc} \psi_t \\ \vphi_t \end{array}\right)&=F(\psi,\vphi)(t),\\
F(\psi,\vphi)(t)&:=
\left(\begin{array}{cc}
\e^{-\i D_A^2 t} \psi_0 - \i  \int_0^t \e^{-\i D_A^2 (t-s)} f_{\psi_s,\vphi_s}  \psi_s\\
\e^{-\i D_A^2 t} \vphi_0 - \i  \int_0^t \e^{-\i D_A^2 (t-s)} g_{\vphi_s, \psi_s} \vphi_s
\end{array}\right),
\end{split}	
\end{equation}
with $f_{\psi_s,\vphi_s}= u^{(1,1)}*|\psi_s|^2 + \frac{1}{R} u^{(1,2)}*|\vphi_s|^2$,  
$ g_{\psi_s,\vphi_s} =  u^{(2,2)}*|\vphi_s|^2 + R u^{(1,2)}*|\psi_s|^2$.
Note that for every $A \in L_{\text{loc}}^2(\R^3)$ the quadratic form $q:C_c^\infty(\R^3)^2\to\C$, $q(\vphi, \psi):=\langle D_A \vphi, D_A \psi \rangle$ defines a positive closable quadratic form because $C_c^\infty(\R^3)$ is dense in $H_A^1(\R^3)$ (see \cite{simon-schroedinger-forms} or \cite[Theorem 7.22]{LL}). Therefore, $D_A^2$ can be realized as the self-adjoint operator associated with the closure of $q$ see \cite[Chapter 10, Theorem 2]{BS}. As a consequence, the unitary group $\e^{-\i D_A^2 t}$ is well-defined and strongly continuous on $L^2(\R^3)$. For the local well-posedness, that is, existence and uniqueness of a solution in the space given by the right hand side of \eqref{eq:CC1} for some $T>0$,  one can argue similarly as in 
\cite{Ca}. The proof is based as usual on the Banach fixed-point theorem. For convenience of the reader, we briefly outline here how $F$ can be proven to be Lipschitz continuous. We have

\begin{equation}\label{eq:difference-nonlinearity}
F(\psi,\vphi)(t)-F(\widetilde{\psi},\widetilde{\vphi})(t)=\left(\begin{array}{cc}
-\i \int_0^t \e^{-\i D_A^2 (t-s)} (f_{\psi_s,\vphi_s}  \psi_s-f_{\widetilde{\psi}_s,\widetilde{\vphi}_s}  \widetilde{\psi}_s)\\
-\i \int_0^t \e^{-\i D_A^2 (t-s)} (g_{\psi_s,\vphi_s}  \vphi_s-g_{\widetilde{\psi}_s,\widetilde{\vphi}_s}  \widetilde{\vphi}_s)
\end{array}\right).
\end{equation}
We drop the time index $s$ for most of the following calculations. 
We start by giving an upper bound on the $H^1_A$-norm of $f_{\psi,\vphi}  \psi-f_{\widetilde{\psi},\widetilde{\vphi}}  \widetilde{\psi}$ by  decomposing  
\begin{displaymath} 
	f_{\psi,\vphi}  \psi-f_{\widetilde{\psi},\widetilde{\vphi}}  \widetilde{\psi}
	= f_{\psi,\vphi} (\psi-\widetilde{\psi})+ (f_{\psi,\vphi}-f_{\widetilde{\psi},\widetilde{\vphi}})  \widetilde{\psi} 	
\end{displaymath}
and we will start with the term $(f_{\psi,\vphi}-f_{\widetilde{\psi},\widetilde{\vphi}})  \widetilde{\psi}$  since  bounding the term $f_{\psi,\vphi} (\psi-\widetilde{\psi})$ in $H^1_A$ is similar and slightly easier.  We start with the rough estimate
\begin{displaymath}
\|(f_{\psi,\vphi}-f_{\widetilde{\psi},\widetilde{\vphi}})  \widetilde{\psi}\|_{L^2} \leq \|f_{\psi,\vphi}-f_{\widetilde{\psi},\widetilde{\vphi}}\|_{L^\infty}      \|\widetilde{\psi}\|_{L^2}\, , 	
\end{displaymath}
and note that the triangle inequality gives 
\begin{align}\label{eq:Linfty}
	 \|f_{\psi,\vphi}-f_{\widetilde{\psi},\widetilde{\vphi}}\|_{L^\infty} \leq  \|u^{(1,1)}*(|\psi|^2- |\widetilde{\psi}|^2)\|_{L^\infty} + \frac{1}{R}\|u^{(1,2)}*(|\vphi|^2- |\widetilde{\vphi}|^2)\|_{L^\infty}. 
\end{align}
 We now bound the first term of the r.h.s. of \eqref{eq:Linfty}, and the last term can be bounded in exactly the same way. The fact that $||\psi|^2- |\widetilde{\psi}|^2| \leq |\psi-\widetilde{\psi}|(|\psi|+|\widetilde{\psi}|)$, together with the Cauchy-Schwarz inequality in the convolution integral, yields
\begin{displaymath}
	\|u^{(1,1)}*(|\psi|^2- |\widetilde{\psi}|^2)\|_{L^\infty} \leq \||u^{(1,1)}|^2*(|\psi- \widetilde{\psi}|^2)\|_{L^\infty}^{1/2} (\|\psi\|_{L^2}+\|\widetilde{\psi}\|_{L^2}).
\end{displaymath}
In addition, \eqref{eq:uij2} gives
\begin{displaymath}
 \||u^{(1,1)}|^2*|\psi- \widetilde{\psi}|^2\|_{L^\infty} \leq 4 (\lambda^{(1,1)})^2 \|D_A(\psi-\widetilde{\psi})\|_{L^2}^2.	
\end{displaymath}
Thus 
\begin{equation}\label{eq:Linfty1}
\begin{split}
	\|f_{\psi,\vphi}-f_{\widetilde{\psi},\widetilde{\vphi}}\|_{L^\infty} 
	& \leq  2|\lambda^{(1,1)}| (\|\psi\|_{L^2} + \|\widetilde{\psi}\|_{L^2}) 
	\|D_A(\psi-\widetilde{\psi})\|_{L^2}  \\
	&\phantom{\leq ~}+ \frac{2|\lambda^{(1,2)}|}{R} (\|\varphi\|_{L^2} + \|\widetilde{\varphi}\|_{L^2}) 
	\|D_A(\varphi-\widetilde{\varphi})\|_{L^2}  
\end{split}
\end{equation}
and therefore 
\begin{equation}
\begin{split}
	\|(f_{\psi,\vphi}-f_{\widetilde{\psi},\widetilde{\vphi}})\widetilde{\psi}\|_{L^2} 
	& \leq  2|\lambda^{(1,1)}| (\|\psi\|_{L^2} + \|\widetilde{\psi}\|_{L^2}) \|\widetilde{\psi}\|_{L^2}
	\|D_A(\psi-\widetilde{\psi})\|_{L^2}  \\
	&\phantom{\leq ~}+ \frac{2|\lambda^{(1,2)}|}{R} (\|\varphi\|_{L^2} + \|\widetilde{\varphi}\|_{L^2}) \|\widetilde{\psi}\|_{L^2}
	\|D_A(\varphi-\widetilde{\varphi})\|_{L^2}  
\end{split}
\end{equation}
Similarly, using \eqref{eq:Linfty1} with $\widetilde{\psi}=\widetilde{\varphi}=0$, one gets 
\begin{displaymath}
\begin{split}
	\|f_{\psi,\vphi} (\psi-\widetilde{\psi})\|_{L^2}
	\le \Big( 2|\lambda^{(1,1)}| \|\psi\|_{L^2}   
	\|D_A \psi \|_{L^2}  + \frac{2|\lambda^{(1,2)}|}{R} \|\varphi\|_{L^2}  
	\|D_A \varphi \|_{L^2} \Big)\|\psi-\widetilde{\psi}\|_{L^2}.
\end{split}  
\end{displaymath}
Putting everything together, we have 
\begin{equation}
\begin{split}
	\| f_{\psi,\vphi}  \psi-f_{\widetilde{\psi},\widetilde{\vphi}}  \widetilde{\psi} \|_{L^2}
	&\le 2|\lambda^{(1,1)}| (\|\psi\|_{L^2} + \|\widetilde{\psi}\|_{L^2}) \|\widetilde{\psi}\|_{L^2}
	\|D_A(\psi-\widetilde{\psi})\|_{L^2}  \\
	&\phantom{\leq ~}+ \frac{2|\lambda^{(1,2)}|}{R} (\|\varphi\|_{L^2} + \|\widetilde{\varphi}\|_{L^2}) \|\widetilde{\psi}\|_{L^2}
	\|D_A(\varphi-\widetilde{\varphi})\|_{L^2}  \\
	&\phantom{\leq ~}+\Big( 2|\lambda^{(1,1)}| \|\psi\|_{L^2}   
	\|D_A \psi \|_{L^2}  + \frac{2|\lambda^{(1,2)}|}{R} \|\varphi\|_{L^2}  
	\|D_A \varphi \|_{L^2} \Big)\|\psi-\widetilde{\psi}\|_{L^2}
\end{split}
\end{equation}

We now need to estimate $\|D_A((f_{\psi,\vphi}-f_{\widetilde{\psi},\widetilde{\vphi}})  \widetilde{\psi})\|_{L^2}$. 
The Leibnitz rule implies
$$D_A((f_{\psi,\vphi}-f_{\widetilde{\psi},\widetilde{\vphi}})  \widetilde{\psi})= (f_{\psi,\vphi}-f_{\widetilde{\psi},\widetilde{\vphi}}) D_A \widetilde{\psi} - \i \widetilde{\psi} (\nabla (f_{\psi,\vphi}-f_{\widetilde{\psi},\widetilde{\vphi}}))$$
The $L^2$-norm of the first summand of the right hand side can be easily bounded using the bound \eqref{eq:Linfty1} on  $\|f_{\psi,\vphi}-f_{\widetilde{\psi},\widetilde{\vphi}}\|_{L^\infty}$. Next we outline how to estimate 
$\|\nabla (f_{\psi,\vphi}-f_{\widetilde{\psi},\widetilde{\vphi}})\|_{L^\infty}.$ 
Note that as $C_c^\infty(\R^3)$ is dense
in $H^1_A(\R^3)$ (see \cite{simon-schroedinger-forms} or \cite[Theorem 7.22]{LL}), we may assume in the following that $\psi, \vphi, \widetilde{\psi}, \widetilde{\vphi} \in C_c^\infty(\R^3)$.
%Kato's equality says $\nabla|\chi|= -\im(\tfrac{\overline{\chi}}{|\chi|}D_A\chi)\in L^2$ for any $\chi\in H^1_A$, where we set $\overline{\chi}/|\chi|=0 $ if $\chi=0$.
 So with $u(x)=\tfrac{e^{-\mu |x|}}{|x|}$, an elementary computation gives for any $\chi\in C_c^\infty(\R^3)$
\begin{align*}
	\nabla\left( u*|\chi|^2 \right) =-2 u*\im(\overline{\chi}D_A\chi).
\end{align*}
 Using this, an elementary computation gives
\begin{align*}
\nabla (f_{\psi,\vphi}-f_{\widetilde{\psi},\widetilde{\vphi}})=& -2   u^{(1,1)} * \Im(\overline{\psi} D_A \psi-\overline{\widetilde{\psi}} D_A \widetilde{\psi}) - \frac{2  }{R} u^{(1,2)} * \Im(\overline{\vphi} D_A \vphi-\overline{\widetilde{\vphi}} D_A \widetilde{\vphi}).
\end{align*} 
To finish the proof we can proceed as above using Cauchy-Schwarz and \eqref{eq:uij2} to see that 
\begin{align*}
	\|\nabla &(f_{\psi,\vphi}-f_{\widetilde{\psi},\widetilde{\vphi}})\|_{L^\infty}
	\le \\
	&C  \left[ (\|D_A\psi\|_{L^2} 
		+ \|D_A\widetilde{\psi}\|_{L^2})\|D_A(\psi -\widetilde{\psi})\|_{L^2}
		+ (\|D_A\vphi\|_{L^2} 
		+ \|D_A\widetilde{\vphi}\|_{L^2})\|D_A(\vphi -\widetilde{\vphi})\|_{L^2}
	\right] 
\end{align*}

Putting the above bounds together shows that there exist a constant $C$, depending only on $|\lambda^{(1,1)}|, |\lambda^{(1,2)}|$ and $R$,  such that 
\begin{equation}\label{eq:Lipschitz}
	\|f_{\psi,\varphi}\psi- f_{\widetilde{\psi},\widetilde{\varphi}}\widetilde{\psi}\|_{H^1_A} 
	\le C \left( \|\psi\|_{H^1_A} + \|\varphi\|_{H^1_A}+\|\widetilde{\psi}\|_{H^1_A} + \|\widetilde{\varphi}\|_{H^1_A} \right)^2 
	\left(  \|\psi-\widetilde{\psi}\|_{H^1_A}+ \|\varphi-\widetilde{\varphi}\|_{H^1_A}\right)
\end{equation}
and a similar bound holds for $\|g_{\psi,\vphi} \vphi- g_{\psi,\vphi}\widetilde{\vphi}\|_{H^1_A}$ with another constant $\widetilde{C}$ depending only on $|\lambda^{(2,2)}|, |\lambda^{(1,2)}|$ and $R$. Thus for the nonlinearity $F$ we get from \eqref{eq:difference-nonlinearity} that  there exist a constant $C$, depending only on $|\lambda^{(1,1)}|, |\lambda^{(1,2)}|, |\lambda^{(2,2)}|$ and $R$,  such that 
\begin{align*}
	\|F(\psi,\varphi)&(t)- F(\widetilde{\psi},\widetilde{\varphi})(t)\|_{H^1_A\times H^1_A} \\
	&\le  Ct\left( \|\psi\|_{H^1_A} + \|\varphi\|_{H^1_A}+\|\widetilde{\psi}\|_{H^1_A} + \|\widetilde{\varphi}\|_{H^1_A} \right)^2 
	\left( \|\psi-\widetilde{\psi}\|_{H^1_A}+ \|\varphi-\widetilde{\varphi}\|_{H^1_A}\right)\, . 
\end{align*}
 
 Having this contraction property, a standard contraction mapping argument in the space $C\left([0,T), H_A^1(\R^3)\times H_A^1(\R^3)\right) $ shows that there is a unique solution in this space, as long as $T$ is small enough. Note that one can show in a standard way the blow up alternative, namely, that if there is a maximal finite $T$ of existence of a unique solution, then $\lim_{t \rightarrow T-} (\|\vphi_t\|_{H^1_A}+\|\psi_t\|_{H^1_A})=\infty$.

\paragraph{Global well-posedness.} A conserved quantity  of the dynamics \eqref{hartree} related to the energy functional $\mathcal{H}_N(\overline{\psi_t},\overline{\vphi_t},\psi_t,\vphi_t)$ of \eqref{totalpde} is
\begin{equation}\label{eq:hndef} 
\begin{split}
 \cE_{mag}(\overline{\psi_t},\overline{\vphi_t},\psi_t,\vphi_t)\;:=\;&R \left(\|D_A\psi_t\|_{L^2}^2+ \frac{1}{2} \scp{\psi_t}{ u^{(1,1)}*|\psi_t|^2\psi_t}\right)   \\
 &+ \frac1R\left(\|D_A\vphi_t\|_{L^2}^2 +\frac{1}{2}\scp{\vphi_t}{ u^{(2,2)}*|\vphi_t|^2\vphi_t}\right)\\
 &+\scp{ \psi_t \otimes  \vphi_t}{ u^{(1,2)}(x_1^{(1)}-x_1^{(2)})  \psi_t \otimes \vphi_t }.
\end{split}
\end{equation}
%Note that we have the equality $\sqrt{N_1N_2}\cE_{mag}(\overline{\psi_0},\overline{\vphi_0},\psi_0,\vphi_0)=\langle \psi_0^{\otimes N_1} \otimes  \vphi_0^{\otimes N_2}, H_N \psi_0^{\otimes N_1} \otimes  \vphi_0^{\otimes N_2} \rangle$.
We will show this by showing that $\cE_{mag}$ is differentiable and $\partial_t\cE_{mag}=0$. The usual problem with this approach is that the solution we have does not have enough regularity to guarantee that the kinetic energy terms  $\|D_A\psi_t\|_{L^2}^2$ and $\|D_A\vphi_t\|_{L^2}^2$ are differentiable. This is simply due to the fact that from the nonlinear Hartee equation we only know $ \partial_t \psi_t= -\i D_A^2\psi_t -f_{\psi_t,\varphi_t}\psi_t\in H^{-1}_A+ H^1_A= H^{-1}_A$ and thus the informal calculation 
\begin{align*} 
	\partial_t \|D_A\psi_t\|_{L^2}^2= 2\re\left(\scp{-\i D_A^2\psi_t}{D_A^2 \psi_t}  +\scp{-\i f_{\psi_t,\varphi_t}\psi_t}{D_A^2\psi}\right) 
	= \re\scp{f_{\psi_t,\varphi_t}\psi_t}{\i D_A^2\psi},
\end{align*}
using that $\langle-\i D_A^2\psi_t, D_A^2 \psi_t \rangle$ is purely imaginary, is not justified because the term $\langle D_A^2\psi_t, D_A^2 \psi_t \rangle$ might be infinite. 

The usual way out is to do an approximation argument, which guarantees some higher regularity of the solutions, so that the above informal calculation becomes correct. However, in this case one either has to establish an existence theory for initial conditions with higher regularity, or invoke some regularization mechanism and control the error. 

Our way out is to work with the interaction picture instead. We rewrite the kinetic energy as 
$\scp{\psi_t}{D_A^2\psi_t}=\scp{D_A\e^{\i D_A^2t}\psi_t}{D_A\e^{\i D_A^2t}\psi_t}$. 
 From \eqref{eq:integral} it follows that
\begin{equation}
\partial_t\left(\e^{\i D_A^2t}\psi_t\right)\;= \;- \i\, \e^{\i D_A^2t} f_{\psi_t, \varphi_t} \psi_t.
\end{equation}
where the right hand side is in $H^1_A(\R^3)$ as follows from \eqref{eq:Lipschitz} by setting $\tilde{\vphi}, \tilde{\psi}=0$. Thus 
\begin{equation}\label{eq:kinder}
\partial_t\scp{\psi_t}{D_A^2\psi_t}=2 \Re \scp{D_A e^{i D_A^2 t} \psi_t}{-\i\, D_A (e^{i D_A^2 t} f_{\psi_t,\vphi_t} \psi_t)}=2 \Re \scp{i D_A  \psi_t}{D_A  (f_{\psi_t,\vphi_t} \psi_t)}.
\end{equation}
We now consider the term
$\scp{ \psi_t \otimes  \vphi_t}{ u^{(1,2)}(x_1^{(1)}-x_1^{(2)})  \psi_t \otimes \vphi_t }$. From the product rule we get 
\begin{equation}\label{eq:derinter1}
\begin{split}
\partial_t\scp{ \psi_t \otimes  \vphi_t}{ u^{(1,2)}(x_1^{(1)}-x_1^{(2)})  \psi_t \otimes \vphi_t }\;=\;&2 \Re \scp{ \partial_t\psi_t}{ u^{(1,2)}*|\vphi_t|^2 \psi_t }\\&+2 \Re \scp{ \partial_t\vphi_t}{ u^{(1,2)}*|\psi_t|^2 \vphi_t },
\end{split}
\end{equation}
where in the right hand side the inner product is interpreted as the canonical pairing of $H^{-1}_A$ with $H^{1}_A$. Observe now that
 \begin{equation}\label{eq:derinter11}
 \begin{split}
 \Re \scp{ \partial_t\psi_t}{ u^{(1,2)}*|\vphi_t|^2 \psi_t }&=\Re \scp{ -\i D_A^2\psi_t-\i f_{\psi_t,\vphi_t}\psi_t}{ u^{(1,2)}*|\vphi_t|^2 \psi_t }\\
 &=\Re \scp{ -\i D_A\psi_t}{ D_A(u^{(1,2)}*|\vphi_t|^2 \psi_t) }
 \end{split}
 \end{equation}
where the last equality follows from the fact that $D_A$ is self-adjoint and from the observation that
$Re \scp{ -\i f_{\psi_t,\vphi_t}\psi_t}{ u^{(1,2)}*|\vphi_t|^2 \psi_t }=0$. Similarly, 
\begin{equation}\label{eq:derinter12}
\Re \scp{ \partial_t\vphi_t}{ u^{(1,2)}*|\psi_t|^2 \vphi_t }=\Re \scp{ -\i D_A\vphi_t}{ D_A(u^{(1,2)}*|\psi_t|^2 \vphi_t) }
\end{equation}
From \eqref{eq:derinter1}, \eqref{eq:derinter11} and \eqref{eq:derinter12} it follows that
\begin{equation}\label{eq:inter111}
\begin{split}
\partial_t\scp{ \psi_t \otimes  \vphi_t}{ u^{(1,2)}(x_1^{(1)}-x_1^{(2)})  \psi_t \otimes \vphi_t }
\;=\;&2 \Re \scp{ -\i D_A\psi_t}{ D_A(u^{(1,2)}*|\vphi_t|^2 \psi_t) } \\ 
&+ 2 \Re \scp{ -\i D_A\vphi_t}{ D_A(u^{(1,2)}*|\psi_t|^2 \vphi_t) }.
\end{split}
\end{equation}
Similarly replacing $u^{(1,2)}$ with $u^{(1,1)}$ and setting $\vphi_t=\psi_t$ in the above equation we obtain that
\begin{equation}\label{eq:inter2}
\begin{split}
\partial_t \scp{\psi_t}{ u^{(1,1)}*|\psi_t|^2\psi_t}&= 
\partial_t \scp{ \psi_t \otimes  \psi_t}{ u^{(1,1)}(x_1^{(1)}-x_2^{(1)})  \psi_t \otimes \psi_t }\\
  &=4 \Re \scp{ -\i D_A\psi_t}{ D_A(u^{(1,1)}*|\psi_t|^2 \psi_t) }.
\end{split}
\end{equation}
From \eqref{eq:kinder} and \eqref{eq:inter2} it follows that
\begin{equation}\label{eq:derkind1}
\partial_t \left(\scp{\psi_t}{D_A^2\psi_t}+\frac{1}{2}\scp{\psi_t}{ u^{(1,1)}*|\psi_t|^2\psi_t}\right)
=\frac{2}{R} \Re \scp{ \i D_A\psi_t}{ D_A(u^{(1,2)}*|\vphi_t|^2 \psi_t) }
\end{equation}
In a similar manner we find
\begin{equation}\label{eq:derkind2}
\partial_t \left(\scp{\psi_t}{D_A^2\vphi_t}+\frac{1}{2}\scp{\varphi_t}{ u^{(2,2)}*|\varphi_t|^2\varphi_t}\right)
=2 R\,  \Re \scp{ \i D_A\vphi_t}{ D_A(u^{(1,2)}*|\psi_t|^2 \vphi_t) }
\end{equation}
From \eqref{eq:hndef}, \eqref{eq:inter111}, \eqref{eq:derkind1} and \eqref{eq:derkind2} 
it follows that 
$\partial_t\cE_{mag}=0$ and so $\cE_{mag}$ is conserved. Similarly an easier computation shows that the masses $\|\psi_t\|_2=\|\vphi_t\|_2=1$ are conserved.

Using the conservation of masses and
\eqref{eq:uij2}, we obtain that
\begin{equation}\label{eq:derunbdd}
\cE_{mag} \geq C_1\left( \|D_A\psi_t\|_2^2+\|D_A\vphi_t\|_2^2\right)-C_2
\end{equation}
for some constants $C_1,C_2>0$, which together with the conservation of energy and the masses eliminates the blow up alternative and therefore global well-posedness
follows.

\subsection{Semi-relativistic case\label{app:sr}}
We will first explain how to obtain global well-posedness in $H^{\frac{1}{2}} \times H^{\frac{1}{2}}$. As the local well-posedness in $H^{\frac{1}{2}} \times H^{\frac{1}{2}}$ is a straightforward rewritting of the arguments in \cite{Le}, where the case of one species of particles is treated, we will only explain how the local well-posedness implies global well-posedness.   Note that Kato's inequality and the initial condition beeing in $H^{\frac{1}{2}} \times H^{\frac{1}{2}}$ ensure that the initial energy is finite. The $L^2$-norms of $\psi_t$ and $\vphi_t$ are also preserved. We will show that there exist constants $C> 0$ such that 
\begin{equation}\label{est:energy}
	 \|(m_1^2-\Delta)^{\frac{1}{4}}\psi_t\|_2^2+\|(m_2^2-\Delta)^{\frac{1}{4}}\vphi_t\|_2^2 \leq C \cE_{sr}(\overline{\psi_t},\overline{\vphi_t},\psi_t,\vphi_t),
\end{equation}
where $\cE_{sr}$ is analogously to $\cE_{mag}$ above defined as
\begin{equation}\label{eq:hnsr} 
\begin{split}
\cE_{sr}(\overline{\psi_t},\overline{\vphi_t},\psi_t,\vphi_t)\;:=\;&R \left(\scp{\psi_t}{  \sqrt{m_1^2-\Delta}\psi_t}+ \frac{1}{2} \scp{\psi_t}{ u^{(1,1)}*|\psi_t|^2\psi_t}\right)   \\
&+ \frac1R\left(\scp{ \vphi_t}{\sqrt{m_2^2-\Delta} \vphi_t} +\frac{1}{2}\scp{\vphi_t}{ u^{(2,2)}*|\vphi_t|^2\vphi_t}\right)\\
&+\scp{ \psi_t \otimes  \vphi_t}{ u^{(1,2)}(x_1^{(1)}-x_1^{(2)})  \psi_t \otimes \vphi_t }.
\end{split}
\end{equation}
The following estimates are completely analogous to those in Appendix \ref{app:self-adjoint}. We define $a_1:= \|(m_1^2-\Delta)^{1/4} \psi_t\|_2$ and $a_2:=\|(m_1^2-\Delta)^{1/4} \vphi_t\|_2$. Then from Kato's inequality $\frac{1}{|x|} \leq \frac{\pi}{2}(-\Delta)^{1/2}$ and the fact that $\|\vphi_t\|_{2} = \| \psi_t \|_{2} =1$, it follows 

\begin{equation*}
\begin{split}
	 \frac{1}{2} \scp{\psi_t}{ u^{(1,1)}*|\psi_t|^2\psi_t}&\geq -\frac{\pi \lambda_-^{(1,1)}}{4} a_1^2,\\
	\frac{1}{2}\scp{\vphi_t}{ u^{(2,2)}*|\vphi_t|^2\vphi_t} &\geq -\frac{\pi\lambda_-^{(2,2)}}{4} a_2^2,\\
	\scp{ \psi_t \otimes  \vphi_t}{ u^{(1,2)}(x_1^{(1)}-x_1^{(2)})  \psi_t \otimes \vphi_t }&\geq-\frac{\pi\lambda_-^{(1,2)}}{2}\min\left(a_1^2,a_2^2\right) \geq -\frac{\pi\lambda_-^{(1,2)}}{2}a_1a_2.
\end{split}	
\end{equation*}
Using completion of the square, we arrive at
\begin{equation*}
\begin{split}
\cE_{sr}\;\geq\;&R \left(1-\frac{\pi \lambda^{(1,1)}_-}{4}\right) a_1^2 + \frac1R\left(1-\frac{\pi\lambda^{(2,2)}_-}{4}\right) a_2^2-\frac{\pi\lambda^{(1,2)}_-}{2}a_1a_2\\
=\;&\frac\pi4\left[\sqrt{R\left(\frac4\pi-\lambda^{(1,1)}_-\right)} a_1-\sqrt{\frac1R\left(\frac4\pi-\lambda^{(2,2)}_-\right)} a_2\right]^2\\
&+\frac\pi2a_1a_2\left[\sqrt{\left(\frac4\pi- \lambda^{(1,1)}_-\right)\left(\frac4\pi-\lambda^{(2,2)}_-\right)}-\lambda^{(1,2)}_-\right].
\end{split}
\end{equation*}
By assumption 	\ref{itm:sr} the last term is well-defined and positive. Arguing as in appendix \ref{app:mag}, one can show that $\cE_{sr}$ is differentiable with derivative zero. As a consequence, the estimate \eqref{est:energy} is satisfied and thus $\|\vphi_t\|_{H^{1/2}}^2 + \|\psi_t\|_{H^{1/2}}^2$ remains bounded. 

The local well-posedness in $H^1 \times H^1$ follows from the standard fixed point arguments outlined in \ref{app:mag} in the case $A=0$. The global well-posedness under our assumptions can be obtained similarly as in Lemma 3 in \cite{Le}: With Gronwall's inequality one can show that  $\|\vphi_t\|_{H^{1}}^2 + \|\psi_t\|_{H^{1}}^2$ grows at most exponentially in $t$ and in particular it cannot blow up in finite time.

\noindent {\bf Acknowledgments.}   Ioannis Anapolitanos and Dirk Hundertmark thank the Deutsche Forschungs\-gemeinschaft (DFG) for financial support through CRC 1173.


\begin{thebibliography}{widestlabel}
	\bibitem[AH]{AH} Anapolitanos, I.; Hott, M.: "A simple proof of convergence to the Hartree dynamics in Sobolev trace norms." Journal of Mathematical Physics 57.12: 122108 (2016)	.
	\bibitem[BS]{BS} Birman, M. S.;  Solomjak, M.Z.: Spectral theory of self-adjoint operators in Hilbert space. D. Reidel Publishing Company (1987).
	
	
	\bibitem[Ca]{Ca} Cao, P.: Global existence and uniqueness for the magnetic Hartree equation. Journal of Evolution Equations 11.4: 811-825 (2011). 
	
  \bibitem[CNPV]{CNPV} Ceccarelli G., Nespolo J., Pelissetto A.,  Vicari E.: Bose-Einstein condensation and critical behavior of two-component bosonic gases. Physical Review A 92(4) (2015).
 
\bibitem[He]{He}  Heil, T.: "Mean-field limits in bosonic systems." (Master's Thesis) LMU Munich (2012).
 
	\bibitem[Her]{Her} Herbst, I. W.: Spectral theory of the operator $(p^2+ m^2)^{1/2}− Ze^2/r$. Communications in Mathematical Physics 53.3: 285-294 (1977).

	\bibitem[KP]{KP} Knowles, A.; Pickl, P.: Mean-Field Dynamics: Singular Potentials and Rate of Convergence. Communications in Mathematical Physics.  Volume 298, Issue 1:   101–138 (2010).
	
	\bibitem[Le]{Le} Lenzmann, E.: Well-posedness for semi-relativistic Hartree equations of critical type. Mathematical Physics, Analysis and Geometry 10.1: 43-64 (2007).
	
	
	\bibitem[LL]{LL} Lieb, E. H.; Loss, M.: Analysis, volume 14 of graduate studies in mathematics. American Mathematical Society, Providence, RI, 4 (2001).
	
	\bibitem[Lu]{Lu} Lührmann, J.: Mean-field quantum dynamics with magnetic fields. Journal of Mathematical Physics 53.2: 022105 (2012).
	
	\bibitem[Mi]{Mi} Michelangeli, A.: "Global well-posedness of the magnetic Hartree equation with non-Strichartz external fields." Nonlinearity 28.8: 2743-2765 (2015)
	
	\bibitem[MMRRI]{MMRRI} Modugno, G.; Modugno, M.; Riboli, F.; Roati, G.; Inguscio, M.: "Two atomic species superfluid." Physical Review Letters, 89(19), 190404 (2002).
	
	\bibitem[MBGCW]{MBGCW} Myatt, C. J.; Burt, E. A.; Ghrist, R. W.; Cornell, E. A.; Wieman, C. E.: "Production of two overlapping Bose-Einstein condensates by sympathetic cooling." Physical Review Letters, 78(4), 586 (1997).
	
	
	\bibitem[MO]{MO} Michelangeli, A., Olgiati, A.: Mean-field quantum dynamics for a mixture of Bose-Einstein condensates. 
	Anal.Math.Phys. 1-40 (2016).
	
		\bibitem[MiS]{MiS} Michelangeli, A.; Schlein, B.: Dynamical collapse of boson stars. Communications in Mathematical Physics 311.3: 645-687 (2012).
		
	\bibitem[MPP]{MPP} Mitrouskas, D.; Petrat, S.; Pickl P.:  Bogoliubov corrections and trace norm convergence for the Hartree dynamics. arXiv:1609.06264 (2016).
	
	
	\bibitem[O]{O}  Olgiati, A.:
	Effective non-linear dynamics of binary condensates and open
	problems. arXiv: 1702.04196 (2017).
	
	\bibitem[Pe]{Pe} Petrat, S.: Derivation of Mean-field Dynamics for Fermions, PhD Thesis (2014).
	
	\bibitem[Pi]{Pi} Pickl, P.: A simple derivation of mean field limits for quantum systems. Letters in Mathematical Physics   97.2: 151-164 (2011).
	
	\bibitem[PS]{PS} Pitaevskii L., Stringari S.: Bose-Einstein Condensation and Superfluidity. Oxford science publications (2016).
		
	
	
	\bibitem[RS1]{RS1} Reed, M.;  Simon, B.: Methods of Modern Mathematical Physics. vol. I: Functional analysis. Academic press New York (1972).
	
		\bibitem[RS2]{RS2} Reed, M.;  Simon, B.: Methods of Modern Mathematical Physics. vol. II: Fourier Analysis, Self-Adjointness. Academic press New York (1975).
	
	\bibitem[RSc]{RSc} Rodnianski, I.; Schlein, B.: Quantum fluctuations and rate of convergence towards mean field dynamics. Communications in Mathematical Physics 291.1: 31--61 (2009).
	
	\bibitem[Si]{simon-schroedinger-forms}
		Simon, B.: Maximal and minimal Schr\"odinger forms.
		Journal of Operator Theory, \textbf{1}, no.\ 1, pp. 37--47 (1979).

	\bibitem[SU]{SU}
	Sims, R.; Ueltschi, D.: Entropy and the Quantum. American Mathematical Society (2009).

\end{thebibliography}
\end{document}